\documentclass[a4paper,UKenglish,cleveref, autoref, thm-restate]{lipics-v2021}

\hideLIPIcs  

\nolinenumbers

\usepackage{graphicx}
\usepackage{todonotes}
\usepackage{amsfonts}  
\usepackage{amssymb}
\usepackage{tikz}
\usetikzlibrary{automata, positioning}
\usepackage{amsmath} 
\usetikzlibrary{calc} 
\usetikzlibrary{shapes}
\usetikzlibrary{shapes.geometric, arrows.meta, positioning}
\usetikzlibrary{arrows.meta}
\usepackage{comment}
\usepackage{color}

\newcounter{rqrmnt}
\newenvironment{requirement}[1][]{\refstepcounter{rqrmnt}\par\smallskip
   \noindent\textbf{R\therqrmnt. #1} \rmfamily}{\smallskip \par} 

\newtheorem{hypothesis*}{Hypothesis}

\title{A Correct by Construction Fault Tolerant Voter for Input Selection of a Control System} 

\author{Arif {Ali AP}}{Department of Computer Science and Engineering, Indian Institute of Technology Palakkad, India}{aparif@gmail.com}{ } {} 

\author{Jasine Babu}{Department of Computer Science and Engineering, Indian Institute of Technology Palakkad, India}{jasine@iitpkd.ac.in}{ } {}

\author{Deepa Sara John}{ISRO Inertial Systems Unit, Indian Space Research Organization, Kerala, India}{deepa\_john@vssc.gov.in}{ } {}

\authorrunning{A.~Ali~AP and J.~Babu and D.~S.~John} 

\Copyright{Arif {Ali AP} and Jasine Babu and Deepa Sara John} 

\ccsdesc[100]{{Theory of computation -> Logic -> Logic and verification}, Theory of computation -> Logic -> Automated reasoning} 

\keywords{Fault Tolerant System Design, Formal Verification, Correct by Construction, Input Selection, Interactive Theorem Proving} 

\category{} 

\relatedversion{} 

\funding{A part of this work is funded by Indian Space Research Organisation (ISRO) RESPOND project RES-IISU-2022-013 titled `Formal Analysis and Verification of Redundancy Management Logic for Navigation Processor used in Man-rated Launch Vehicle'.}

\acknowledgements{The authors acknowledge Saarang S and Mis Ab VP, former undergraduate students of IIT Palakkad, for their involvement in preliminary experiments with the Coq implementation of the voter, and Sandra S, visiting student at IIT Palakkad, for testing support. The authors also thank Murali Krishnan, Professor, National Institute of Technology Calicut, for his useful feedback on the initial draft.}

\begin{document}

\maketitle

\begin{abstract}
Safety-critical systems use redundant input units to improve their reliability and fault tolerance.  A voting logic is then used to select a reliable input from the redundant sources.  A fault detection and isolation rules help in selecting input units that can participate in voting. This work deals with the formal requirement formulation, design, verification and synthesis of a generic voting unit for an $N$-modular redundant measurement system used for control applications in avionics systems. The work follows a correct-by-construction approach, using the Rocq theorem prover.
\end{abstract}
\section{Introduction}
Redundancy management is used in safety-critical systems to improve their reliability and fault tolerance~\cite{divito1990formal,goldberg1984,wakerly1978synchronization}. Examples of some well known domains of safety critical applications include avionics control, automated vehicles and nuclear reactor control systems ~\cite{Knight2002Safety,Kassab2014Anovel}. 
Since the reliability guarantee provided by testing is insufficient, such systems are suitable candidates for formal verification~\cite{HandbookofModelchecking}. From the late 1980s, usage of formal verification for the design of fault-tolerant systems~\cite{Aviziens1976Fault} was supported by agencies such as NASA, which work in the domain of safety-critical applications~\cite{rushby1989formal,divito1990formal,Rushby1993Formal}. Owre et al.~\cite{Owre1995Formal} mention the formal verification of fault-tolerant architectures as a major motivation for the development of the PVS theorem prover.  Interactive consistency and convergence~\cite{lincoln1993formal}, clock synchronization~\cite{rushby1994formally,miner1993verification}, and redundancy management~\cite{divito1990formal} are some well-studied subproblems of fault-tolerant system design for which formal verification has been attempted. There are also works on the verification of timing specifications of fault-tolerant systems~\cite{rushby1999systematic,rushby2002overview,jones2017modular,saha2016formal}. Most of the works on formal verification of fault tolerant systems~\cite{lincoln1993formal,rushby1994formally,dajani2003formal,Pike2006Anote,krishnan2024formalclk,krishnan2024formal} discuss only verification of the formal model and not the code synthesis. There is a recent work by Rahli et al.~\cite{rahli2018velisarios} that discusses Rocq based synthesis of a code implementing a Byzantine fault tolerant  state machine replication protocol based on message passing.

Our work deals with the formal modelling, verification and construction of an input selection unit for a synchronous $N$-modular redundant measurement system used for control applications in avionics systems. Though the redundant input units work synchronously and measure the same physical quantity, there can be differences in the measured value due to the noises in the measurement system (see Section~\ref{sec:unitfaultmodel}). Moreover, some of these input units may have transient faults due to factors such as cosmic radiations~\cite{Softerrorbook2011} and vibrations~\cite{Hitt2006Fault}, or permanent faults due to hardware failures~\cite{Davis1987Davis}. Hence, it becomes necessary to apply an input selection logic to the values obtained from the redundant input units and produce a single output value corresponding to the physical quantity. In this work, we follow a correct-by-construction approach for synthesis of the input selection unit which uses a noise-resilient voting algorithm for the selection. Henceforth, we refer to the input selection unit as the \textit{voter unit}.

The interesting application scenarios are when the system is running continuously, where in every cycle of operation, each input unit produces a measurement. In many real scenarios, the incremental change in the physical quantity being measured (such as distance to an object, angular velocity, atmospheric pressure, etc. measured in an avionics control unit) from one cycle to the next is small and smooth.  In each cycle, the voter unit has to select one of the input units as a \textit{prime unit}. The selection should ensure a reliable output value by discarding the measurements having transient errors. In addition, permanently faulty units should be identified and isolated, so that they do not corrupt the feed to the controller. To do this effectively, the input selection of each cycle has to also take into account some information from the past cycles. If the measurement domain is $\mathcal{D}$, the number of input units is $N$ and the information based on the past $t$ cycles is used for the calculations, then the voting unit essentially calculates a function $f: \mathcal{D}^{tN} \mapsto \mathcal{D}$ in every cycle. In a more general situation, some self-reported parameters from the input units are also required for the decision making. It is also desired that, along with the selected measurement value, some reliability parameters of the output are also supplied to the controller. The reader may refer to Figure~\ref{fig:syssetup} which depicts an illustration of such a system. The formal verification of such systems becomes challenging due to the large size of the underlying state space~\cite[Section 1.3.1]{HandbookofModelchecking}.

\begin{figure}
\centering
 \resizebox{\textwidth}{!}{\begin{tikzpicture}[node distance=2.5cm, auto]
\tikzstyle{block} = [rectangle, draw, fill=gray!20, text centered, minimum height=3em, minimum width=3.2cm, minimum height=1.75cm]
\tikzstyle{arrow} = [thick,->,>=Stealth]

\node (source) [cloud, cloud puffs = 10, draw, minimum width = 3.5cm, minimum height = 2cm, fill = gray!10] at (-7,-1){\LARGE Data};
\node (source2) [cloud, cloud puffs = 10, draw, minimum width = 3.5cm, minimum height = 2cm, fill = gray!10] at (-11,-1){\LARGE Data};
\node (input1) [block] at (-2.4,1) {\LARGE Input Unit 1};
\node (inputn) [block] at (-2.4,-4.2) {\LARGE Input Unit N};
\node (voter) [block] at (7.1,-2) {\LARGE Input Selection Unit};
\node (controller) [block] at (18,-2) {\LARGE Controller};


\draw [arrow] (source) -- (input1.west);
\draw [arrow] (source) -- (inputn.west);
\draw [arrow] (source) -- (-3.4,-1);

\draw [arrow] (input1) -- (voter) 
node[pos=0.4, sloped]{\parbox{4cm}{\LARGE(unit status info,\\ measured value)}}(voter);


\draw [arrow] (inputn) -- node[pos=0.4, sloped] 
{\parbox{4cm}{\LARGE\raggedleft(unit status info,\\ measured value)}}(voter.west);

\draw [arrow] (voter) -- node[pos=0.5]
{\parbox{5.5cm}{\LARGE(output value,\\
reliability parameters)}} (controller);

\draw [arrow] (voter.east) -- (controller.west);

\path (input1) -- (inputn) node [black, font=\Huge, pos=0.5, sloped] {~~$\dots$ $\dots$~~};

\draw [arrow]  (voter) -- (11,-2.5);
\draw [arrow]  (11,-2.5) -- (11, -5);
\draw [arrow]  (11, -5) -- node[above] {\LARGE preserved data} (3.3, -5);
\draw [arrow]  (3.3, -5) -- (3.3,-2.8);
\draw [arrow]  (3.3, -2.8) -- (voter);

\end{tikzpicture}}
 \caption{In each cycle, the input selection unit has to select an appropriate measured value and feed it to the controller, along with some parameters indicating the reliability of the output. Between consecutive cycles, the change in the physical quantity being measured is assumed to be small. Preserved data stores some information from past cycles and it is used for fault handling and output generation. The output value could be retained from the previous cycle, when the current cycle measurement is not reliable.}
 \label{fig:syssetup}
\end{figure}
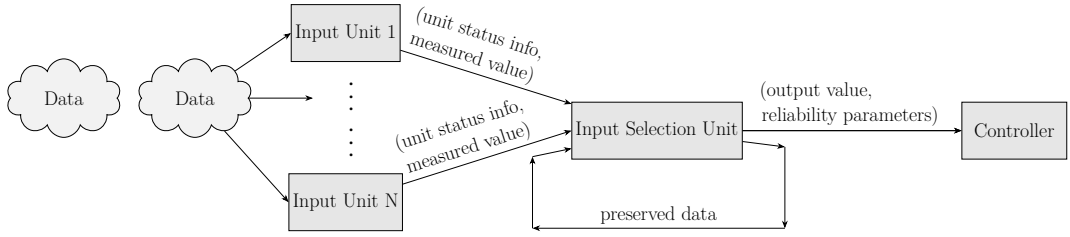

One aspect of transient resilience of the system is to ensure that measurements with transient errors are discarded, without getting fed to the controller. On the other hand, measurements are usually taken from redundant units that are physically separated and there could be minor deviations in the measurements made by them. Hence, it is desirable to avoid \textit{frequent switching} of the prime unit due to transient failures, to maintain better stability. Although these two aspects of transient resilience may seem conflicting, the assumption that the  incremental variations in the measured quantity between consecutive cycles is small, makes it possible to strike a balance.  If the prime unit selected for the previous cycle has a transient error in the current cycle, the algorithm can retain the output value of the previous cycle without switching the prime unit, as long as the age of the measurement is within the allowed limits. If an input unit is showing failure for many consecutive cycles beyond a \textit{persistence limit}, then it has to be permanently isolated.
These requirements make the design of an input selection unit for a controller different from a typical voting algorithm.

As explained by Butler~\cite{Butler1993Formal} and Miller et al.~\cite{Miller2003Proving}, a major effort in the formal verification of a safety critical system goes into unambiguously capturing the requirements of the system using formal logic.  A major contribution of this work is the formulation of a consistent set of requirements relevant to the application context presented above.  We demonstrate the theoretical limits of achievable soundness and completeness. Our voter unit algorithm is proven to match the theoretical limits. 

If the number of input units without a permanent fault falls below a certain threshold, then it is well known that fault identification cannot achieve completeness~\cite{wakerly1978synchronization}. However, by fine-tuning the fault identification algorithm, in certain situations it is possible to identify with certainty, a subset of units as non-faulty and another subset of units as faulty. In such cases as well, we derive the theoretical limits of fault identification and our algorithm matches these limits. Such a finer fault identification will help in prolonging the duration of reliable operation of the voter unit, by avoiding  mis-classification of a healthy input unit as permanently faulty. 

When the algorithm is designed to prolong the duration of operation of the voter unit this way, there could be certain boundary cases wherein the prime unit of the previous cycle got isolated due to a permanent fault and it is not theoretically possible to identify with certainty even one unit as non-faulty to assign as the new prime unit. In such situations, if the number of units without a permanent fault is above \textit{the minimum threshold for continuing the operation} of the voter unit, then for a short while our algorithm will change the state of the voter unit as un\_id (unidentifiable), in anticipation of a recovery from transient errors, while continuing to provide the output from the previous cycle. The algorithm will take care of bounding the age of the output in such scenarios as well. 

To make the system description and its requirements more precise, it is necessary to formally specify the assumed fault model and fault hypotheses~\cite{rushby2002overview} indicating the type, number and rate of faults that the system can tolerate while remaining operational~\cite{lamport1982byzantine,Powell1992Failuremode}. This is done in Section~\ref{sec:systemdescription}. In Section~\ref{sec:functionalrequirements}, we formulate the functional requirements of the system.  In Section~\ref{sec:design}, we give an overall design of the voter unit. The voter unit is designed as a state transition system with three major functions: fault identification, handling of transient and permanent failures using a multi cycle history of faults, and generation of a reliable output, which will be fed to the controller. 

Since theorem proving techniques are known to be more scalable compared to other approaches to formal verification~\cite{HandbookofModelchecking} for systems with a large state space, this work is done using the Rocq interactive theorem prover~\cite{RocqWebcite}. 
Rocq theorem prover is based on the calculus of inductive construction and satisfies the de Bruijn criterion~\cite{Chlipala2013certified}. In addition, Rocq supports higher-order logics and dependent types, which makes it easy to express complex specifications succinctly. We have attempted to follow Rocq's design principle of defining functions carrying proof of its properties embedded with it, rather than completing a code first and then proving the properties. Some of the requirements are shown to be satisfied as invariants of the states, while the remaining requirements, which depend on the multi cycle behaviour of the system, are captured as invariants maintained during every state transition. 
To allow the designer to fine-tune the reliability parameters to match the application specific requirements, we maintain the definition of some basic reliability parameters in a separate configuration module, and the system implementation is parameterized by them. The code of the voter unit is synthesized with the selected values of these parameters (see Figure~\ref{fig:voter_generator}).

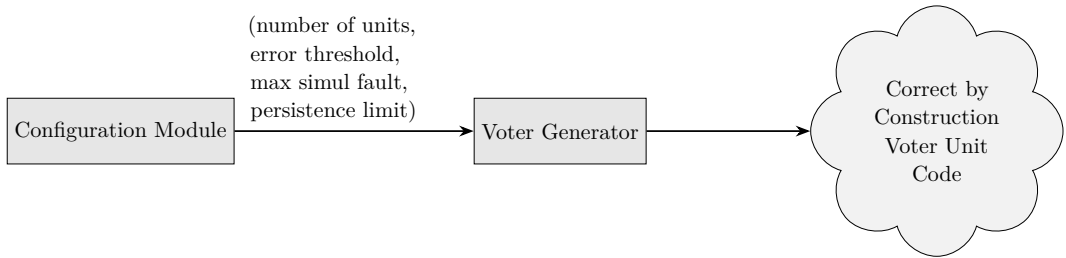
\begin{figure}
\centering
 \resizebox{\textwidth}{!}{\begin{tikzpicture}[node distance=2.5cm, auto]
\tikzstyle{block} = [rectangle, draw, fill=gray!20, text centered, minimum height=3em, minimum width=2cm]
\tikzstyle{arrow} = [thick,->,>=Stealth]

\node (config) [block] at (-3,1) {Configuration Module};
\node (votergen) [block] at (4,1) {Voter Generator};

\node (code) [cloud, cloud puffs = 8, draw, minimum width = 2cm, minimum height = 2cm, fill = gray!10] at (10,1){\parbox{2.1cm}{\centering Correct by\\ Construction\\ Voter Unit Code}};

\draw [arrow] (config) -- (votergen) 
node[midway]{\parbox{3.3cm}{(number of units,\\ error threshold,\\ max simul fault,\\persistence limit)}};

\draw [arrow] (votergen) -- (code);

\end{tikzpicture}}
\caption{To adjust the reliability levels, some parameters are kept as designer configurable, with a permissible range of values for each. A voter unit code with the specified set of parameter values can be generated.}
 \label{fig:voter_generator}
\end{figure}

\section{System Description and Fault Model}\label{sec:systemdescription}
In this section, we formally define the input and output formats of the voter unit and the fault model and fault hypothesis. 
\subsection{Input Format of the Voter Unit}\label{sec:unitmodel}
In each cycle of operation, the voter receives input from $N$ redundant input units that work synchronously and measure the same physical quantity that continuously varies over time. The value of $N$ can be decided by the user. The voter unit maintains a unique identification number \texttt{uid} for each input unit based on the port of the voter unit interface to which the input unit is connected. From each input unit, the voter receives a \texttt{reading}, which is a 2-tuple consisting of a) the measured value \texttt{val}, and b) the self-identified health status of the unit \texttt{hw\_hlth}. This kind of data format where a measurement is available with a health status is commonly used by sensors and measurement units used for critical applications~\cite{ADIS1675datasheet, GG1320usermanual}. 

The measurement from a properly functioning input unit may  be noisy, but within an acceptable threshold $\delta$ from the ground truth. An input unit can malfunction due to a transient or permanent fault. 

\subsection{Fault Model}\label{sec:unitfaultmodel}
As described in the input format, in each cycle, each input unit has a self-identified health status: \texttt{good} or \texttt{bad}. Another type of fault that is not self-identifying is a deviation greater than a fixed noise threshold $\delta$ from the ground truth of the physical quantity being measured. We call it a \textit{deviation fault}. A faulty behaviour of an input unit is defined as an occurrence of any one of these two faults. 
A faulty behaviour is considered \textit{transient} if it persists less than a  consecutive number of cycles of operation, denoted by \texttt{persistence\_lmt}.  Otherwise, it is considered a \textit{permanent fault}. For transient fault resilience, \texttt{persistence\_lmt} should be greater than $1$. It is assumed that the faulty behaviour is limited to the input units alone, and the hardware interface of the voter unit receiving readings from the input units is assumed to be non-faulty. Byzantine faults~\cite{lamport1982byzantine} are not considered in our formulation. 
We need to make a \textit{simultaneous fault hypothesis} about the maximum number of simultaneous new faults that are anticipated in a cycle of operation. This assumption is stated below.
\begin{hypothesis*}[Simultaneous Fault Hypothesis]
   In any cycle, among the input units that do not have a permanent fault before the cycle begins, at most $\mathtt{max\_simul\_fault}$ units can show faulty behaviour.
\end{hypothesis*}
The correctness of the voter unit's behaviour can be guaranteed only if a certain minimum number of input units without a permanent fault are present. Formally, we define a parameter \texttt{min\_required} and say that the system satisfies 
\textit{maximum permanent fault assumption} when there are at least \texttt{min\_required} number of input units without a permanent fault.  
By \textit{simultaneous fault hypothesis}, it is evident that \texttt{min\_required} should be at least $\mathtt{max\_simul\_fault} +1$. 
The values of parameters $\delta$, \texttt{persistence\_lmt} and \texttt{max\_simul\_fault} can be decided by the user.

\subsection{Output Format of the Voter Unit}\label{sec:outputformat}  
The format of the output produced by the voter unit is a 4-tuple consisting of \texttt{voter\_output}, \texttt{output\_age}, \texttt{validity\_status} and \texttt{presrvd\_data}. The first three are for feeding to the controller, while \texttt{presrvd\_data} is a feedback to the voter unit itself for decision making during state transitions.

In each cycle, the voter should select one of the input units as the \textit{prime unit}. The \texttt{voter\_output} consists of the \texttt{uid} of the prime unit and a \texttt{reading} from the prime unit. The \texttt{voter\_output} is a current cycle measurement of the prime unit if the prime unit has no identified fault in the current cycle. Otherwise, the previous cycle \texttt{voter\_output} is retained. The \texttt{output\_age} is used to indicate the number of cycles since when the voter output measurement has not been updated. The \texttt{validity\_status} is an indication of the reliability of the \texttt{voter\_output} and 
can be \texttt{valid}, \texttt{un\_id} (unidentifiable) or \texttt{not\_valid}. The \texttt{presrvd\_data} consists of the \texttt{uid}, \texttt{reading} and accumulated fault history of the prime input corresponding to the \texttt{voter\_output}.

\section{Formulation of Functional Requirements}\label{sec:functionalrequirements}
The functional requirements of the generic data selection and voting unit is described in this section.  
The requirements are broadly classified into three categories: (i) related to fault identification, (ii) related to isolation of permanently faulty units and (iii) related to output generation. We provide an analysis of the theoretical limits achievable for fault identification and show that the requirements formulated are tight with respect to these limits. The proof of correctness of fault identification rules presented in a conventional way in this section are machine verified in our Rocq implementation.  
The consistency of the set of requirements is verified in Rocq, and the code of the voter unit produced is proven to be correct by construction satisfying all the requirements. 
 
\subsection{Fault Identification}
In each cycle, the voter unit should analyze the input received from the input units and record the information of faults identified by the fault identification algorithm. As a method of recording the faults, the voter unit will maintain a  \texttt{risky\_count} corresponding to each input unit, which indicates the number of consecutive cycles of operation (including the current cycle) in which the measurement from the unit is not identified as non-faulty.
Before explaining the details of requirements related to fault identification, we need to first give the definitions of soundness and completeness of fault identification. 
\subsubsection{Soundness and completeness}
As the actual value (ground truth) of the physical quantity being measured is unknown, units with deviation faults have to be identified using mutual deviation checks. As per the fault model, a non-faulty measurement value can have a deviation of $\delta$ from the ground truth, in a positive or negative direction. Hence, a deviation of $2\delta$ between two inputs is allowed and cannot be identified as a fault. 
\begin{definition}[miscomparison]
    A deviation of more than $2\delta$ between measurements in a cycle from a pair of input units is a \textit{miscomparison}. 
\end{definition}
A miscomparison indicates that at least one of the two measurements involved in the comparison is faulty. Since it is theoretically impossible to identify all deviation faults in certain scenarios~\cite{wakerly1978synchronization}, we want to classify the units to have their \texttt{miscomparison\_status} to be one among \texttt{miscomparing}, \texttt{not\_miscomparing} and \texttt{maybe\_miscomparing}. 

We start by analysing a scenario of fault identification that will justify the definitions of soundness and completeness of fault identification that follow. 
\begin{claim}\label{clm:justification_incompleteness}
    There are situations where the measurement of an input unit $u$ deviates $3\delta$ from the ground truth and it is still impossible to detect the unit as faulty. 
\end{claim}
\begin{proof}
Suppose that unit $u$ gives a measurement $x+3\delta$, while other units give a measurement value $x+\delta$. Now, consider two scenarios (i)  the (unknown) ground truth is $x$ and (ii) the ground truth is $x+2\delta$. 

 In both scenarios, the mutual deviation check will not show any miscomparison. In the first case, the measurement of $u$ deviates $3\delta$ from the ground truth and in the second case, its deviation is within the noise threshold $\delta$ from the ground truth. Since the ground truth is unknown, it is impossible to distinguish between these two scenarios.
\end{proof}
 The above claim justifies the definitions of soundness and completeness of the fault identification stated below. We will discuss the acheivability of these in Section~\ref{sec:soundandcomplete}. 
\begin{definition}[Soundness and Completeness of Deviation Fault Identification]\label{def:sound_complete}
~
\begin{itemize}
    \item (Soundness a) If the $\mathtt{miscomparison\_status}$ of a unit is $\mathtt{miscomparing}$, its measurement deviates more than $\delta$ from the (unknown) ground truth.  
    \item (Soundness b) If the $\mathtt{miscomparison\_status}$ of a unit is $\mathtt{not\_miscomparing}$, then the deviation of the measurement of that unit from the ground truth is at most $3\delta$.
    \item (Completeness a) If the measurement of a unit deviates more than $3\delta$ from the ground truth, then its  $\mathtt{miscomparison\_status}$ should be $\mathtt{miscomparing}$. 
    \item (Completeness b) If the deviation of the measurement of a unit from the ground truth is at most $\delta$, its $\mathtt{miscomparison\_status}$ should be $\mathtt{not\_miscomparing}$. 
\end{itemize}
\end{definition}
\subsubsection{Requirements of soundness and completeness of fault identification}\label{sec:soundandcomplete}
The requirement regarding the handling of self-identifying faults is straightforward. Input units with self-identifying \texttt{ bad} health must be identified as faulty. The following basic correctness requirement for fault identification algorithm is the following.
\begin{requirement}\label{risky_count_correctness}
For an input unit which is not yet marked as permanently faulty, 
 \begin{itemize}
     \item \texttt{risky\_count} is either zero or one more than its \texttt{risky\_count} in the previous cycle and
     \item \texttt{risky\_count} is incremented in a cycle if and only if its self identifying health status is \texttt{bad} or its \texttt{miscomparison\_status} is \texttt{miscomparing} or \texttt{maybe\_miscomparing}.
 \end{itemize}
\end{requirement}

Now, we will identify some conditions under which the completeness property of fault identification cannot be achieved, so that the requirements on the completeness can be fine tuned. One such scenario is when among the input units that do not have a permanent fault, there are not enough units without self-identifying faults. 
We will make this statement more precise shortly. Among input units without a permanent fault, if there are $k$ units with self-identifying  \texttt{bad} health in a cycle, then by \textit{simultaneous fault hypothesis}, there can be at most $\mathtt{max\_simul\_fault}-k$ units with non-self-identifying deviation faults. Let $\mathtt{mis\_flt\_lmt}=\mathtt{max\_simul\_fault}-k$. 
\begin{claim}\label{clm:insufficient_units}
    There are situations in which the number of input units that do not have a permanent fault and having good health is $2*\mathtt{mis\_flt\_lmt}$ and yet it is not possible to achieve completeness of deviation fault identification.
\end{claim}
\begin{proof}
    Consider a situation where $\mathtt{mis\_flt\_lmt}$ is $2$ and there are $4$ units $u_1, u_2, u_3, u_4$ without a permanent fault yet and with \texttt{good} health. Suppose $u_1$ and $u_2$ have the same measurement $x$ and $u_3$ and $u_4$ have the same measurement $y$, which is more than $3\delta$ away from $x$. There are two possible scenarios: (i) both $u_1$ and $u_2$ are faulty, while $y$ is within $\delta$ of the ground truth (ii) both $u_3$ and $u_4$ are faulty, while $x$ is within $\delta$ of the ground truth. 

    Here, since $|x-y| > 3\delta$, it is clear that either $x$ or $y$ deviates more than $\delta$ from the ground truth. However, since the ground truth is unknown, it is impossible to distinguish which among these actually occurred. Hence, even when the number of input units that do not have a permanent fault and having good health is $2*\mathtt{mis\_flt\_lmt}$, completeness cannot be achieved.
\end{proof}
\begin{definition}[minimum surviving units assumption]\label{def:min_surviving_units}
  The system is said to satisfy the \textit{minimum surviving units assumption} when among input units that do not have a permanent fault in a cycle, at least $2*\mathtt{mis\_flt\_lmt}+1$ units have \texttt{good} health status. 
\end{definition}

 
By Claim~\ref{clm:insufficient_units}, the following requirement regarding completeness of fault identification is tight with respect to the theoretical limit. 
\begin{requirement}\label{maybe_nil}
    When the system satisfies the minimum surviving units assumption, the\\ \texttt{miscomparison\_status} of each unit should be either \texttt{miscomparing} or \texttt{not\_miscomparing} and this classification should satisfy the completeness requirements given in Definition~\ref{def:sound_complete}. 
\end{requirement} 
It can be shown that this requirement is always satisfiable. We will address this point later. The most important requirement of our deviation fault identification algorithm is the following. 
\begin{requirement}\label{always_sound}
 The soundness conditions specified in Definition~\ref{def:sound_complete} must always be satisfied by the fault identification algorithm.  
\end{requirement}

\subsubsection{Avoiding mis-classifications to extend operational life}\label{sec:prolongisolation}

If the minimum surviving units assumption is not satisfied, our main objective is to satisfy the soundness requirements given in Definition~\ref{def:sound_complete}. For this, it suffices to trivially declare all units as \texttt{maybe\_miscomparing} when the \textit{minimum surviving units assumption} is not satisfied. However, this may lead to an early declaration of the input units as permanently faulty and may cause the system to fail to satisfy the \textit{maximum permanent fault assumption}. Hence, let us explore some more scenarios in which partial fault identification can be achieved. 
\begin{claim}\label{clm:partial_identification_miscomparing}
   If the measurement value of an input unit $u$ miscompares with at least $\mathtt{mis\_flt\_lmt}+1$ other input units without a permanent fault yet and with \texttt{good} health, then the measurement of $u$ deviates more than $\delta$ from the ground truth. Further, there are situations, where a unit is miscomparing with \texttt{mis\_flt\_lmt} other units and still the unit cannot be identified as faulty. 
\end{claim}
\begin{proof}
Among the $\mathtt{mis\_flt\_lmt}+1$ other units with which unit $u$ miscompares, by the \textit{simultaneous fault hypothesis}, there must be at least one non-faulty unit $u'$. Since measurement of $u'$ is within $\delta$ from ground truth and the measurements of $u$ and $u'$ deviate more than $2 \delta$, it follows that the measurement of $u$ definitely deviates more than $\delta$ from the ground truth. 
Hence, $u$ is faulty. To see the second part of the claim, it suffices to consider the scenario described in the proof of Claim~\ref{clm:insufficient_units}.
\end{proof}
 Now, let us consider another interesting scenario where a partial fault identification is possible. To simplify our presentation, we make the following definitions. 
\begin{definition}[$\mathtt{miscomparing\_list}$, $\mathtt{rem\_mis\_flt\_lmt}$]
$\mathtt{miscomparing\_list}$ in a cycle is the list of  input units without a permanent fault yet and with $\mathtt{good}$ health and whose measurement miscompares with at least $\mathtt{mis\_flt\_lmt}+1$ other such units. We define $\mathtt{rem\_mis\_flt\_lmt} = \mathtt{mis\_flt\_lmt} - |\mathtt{miscomparing\_list}|$. 
\end{definition}
\begin{definition}[identified fault]
    A unit is said to have an identified fault in a cycle if its health status is $\mathtt{bad}$ or it is in the $\mathtt{miscomparing\_list}$. 
\end{definition}

\begin{claim}\label{clm:clm:partial_identification_not_miscomparing}
Let $u$ be a not yet permanently faulty input unit without an identified fault. 
    If the deviation of the measurement of $u$  is within $2\delta$ from  \texttt{rem\_mis\_flt\_lmt} other such units, then the measurement of $u$ is within $3\delta$ of the ground truth. Further, there are situations where the measurement of such a unit $u$ is within $2\delta$ from exactly \texttt{rem\_mis\_flt\_lmt} $-1$ other units and still it is impossible to assign a miscomparison status \texttt{miscomparing} or \texttt{not\_miscomparing} to  $u$ satisfying the soundness requirements. 
\end{claim}
\begin{proof}
By simultaneous fault hypothesis, \texttt{rem\_mis\_flt\_lmt} is an upper bound on the remaining number of units with faults among units with health status \texttt{good}, after applying the rule in Claim~\ref{clm:partial_identification_miscomparing}.  Hence, if the deviation of the measurement of a unit $u$ which does not have an \textit{identified} fault is within $2\delta$ from  \texttt{rem\_mis\_flt\_lmt} other such units, then we can conclude that either $u$ or at least one among them is non-faulty. In both these cases, we can conclude that the unit $u$
has a measurement within $3\delta$ from the ground truth. 

To see the second part of this claim, it is sufficient to analyze a similar situation as given in the proof of Claim~\ref{clm:insufficient_units}.  Suppose  $\mathtt{mis\_flt\_lmt}$ is $2$ and there are $4$ units $u_1, u_2, u_3, u_4$ without a permanent fault yet and with \texttt{good} health. Suppose $u_1$ and $u_2$ have the same measurement $x$ and $u_3$ and $u_4$ have the same measurement $y$, which is more than $3\delta$ away from $x$. Consider two possible scenarios: (i) the ground truth is $x$ (ii) the ground truth is $y$. Here, observe that \texttt{rem\_mis\_flt\_lmt} is the same as \texttt{mis\_flt\_lmt}, which is $2$. To satisfy the soundness condition, in the first scenario, miscomparison status of neither $u_1$ nor $u_2$ can be assigned as \texttt{miscomparing}. In the second scenario, their miscomparison  status cannot be assigned as \texttt{not\_miscomparing}. The situation of $u_3$ and $u_4$ is symmetric. 

Since it is impossible to distinguish between the two scenarios without knowing the ground truth, the miscomparison status of none of the units can be set to  \texttt{miscomparing} or \texttt{not\_miscomparing} ensuring the soundness requirements.
\end{proof}
Claim~\ref{clm:partial_identification_miscomparing} and Claim~\ref{clm:clm:partial_identification_not_miscomparing} show that the two requirements stated below are tight with respect to the theoretical limits. 

\begin{requirement}\label{miscomparing} If the measurement value of an input unit miscompares with at least $\mathtt{mis\_flt\_lmt}+1$ other input units without a permanent fault yet and with \texttt{good} health, then it should have \texttt{miscomparison\_status} as  \texttt{miscomparing}.
\end{requirement}

\begin{requirement} \label{not_miscomparing} 
Among the not permanently faulty input units which do not have an identified fault, if the measurement of a unit is within $2\delta$ from at least \texttt{rem\_mis\_flt\_lmt}
others, then its \texttt{miscomparison\_status} should be \texttt{not\_miscomparing}.
\end{requirement}

Now, the following is an interesting observation.
\begin{proposition}\label{prop:completeness_achieved}
When the system satisfies the minimum surviving units assumption, requirements R\ref{miscomparing} and R\ref{not_miscomparing} are sufficient to guarantee requirements R\ref{maybe_nil} and R\ref{always_sound}.
\end{proposition}
\begin{proof}
Suppose that the system satisfies the minimum surviving units assumption. Then $2*\mathtt{mis\_flt\_lmt}+1$ input units are available among units without a permanent fault yet and with \texttt{good} health. Consider such an arbitrary unit $u$.  
Among the other $2*\mathtt{mis\_flt\_lmt}$ units other than $u$, if there are at least $\mathtt{mis\_flt\_lmt}+1$ units with which $u$ miscompares, then by Claim~\ref{clm:partial_identification_miscomparing}, measurement of $u$ deviates more than $\delta$ from the ground truth. Otherwise, there are at most $\mathtt{mis\_flt\_lmt}$ other units with which $u$ miscompares and hence there are at least $\mathtt{mis\_flt\_lmt} \ge \mathtt{rem\_mis\_flt}$ units whose measurement is within $2\delta$ from that of $u$. In this case, by Claim~\ref{clm:clm:partial_identification_not_miscomparing}, the measurement of $u$ is within $3 \delta$ deviation from the ground truth. In the first case, by R\ref{miscomparing} the miscomparison status of $u$ will be \texttt{miscomparing} and the second case, by R\ref{not_miscomparing} $u$ will have miscomparison status \texttt{not\_miscomparing}. In both cases, the soundness and completeness properties are satisfied. Since this can be done for an arbitrary $u$, R\ref{maybe_nil} will be satisfied. Since soundness is ensured, R\ref{always_sound} also holds. 
\end{proof}
\begin{note}
 Requirements R\ref{maybe_nil} and R\ref{always_sound}
 together means that under minimum surviving units assumption, all units are assigned miscomparison status \texttt{miscomparing} or \texttt{not\_miscomparing} and the classification is sound and complete. By Proposition~\ref{prop:completeness_achieved}, this can be ensured by satisfying requirements R\ref{miscomparing} and R\ref{not_miscomparing}. Moreover, we can algorithmically ensure R\ref{miscomparing} and R\ref{not_miscomparing} by applying the straightforward conditions in 
Claim~\ref{clm:partial_identification_miscomparing} and Claim~\ref{clm:clm:partial_identification_not_miscomparing}.  
  \end{note}
  
\subsection{Isolation of Permanently Faulty Units}
As discussed in the previous section, in each cycle the voter unit uses its fault identification algorithm to identify the fault status of each input unit. The voter unit maintains some information to keep track of the accumulated fault status of each unit. This information will be used to distinguish between transient faults and permanent faults in order to avoid corruption of the feed to the controller in future from a permanently faulty unit.

One such information maintained is the \texttt{risky\_count} of an input unit, which needs to be updated as mentioned in requirement R\ref{risky_count_correctness}. 
An input unit with a non-zero \texttt{risky\_count} which is less than \texttt{persistence\_lmt} is regarded to have a \textit{transient fault}. If the transient fault clears before it reaches \texttt{persistence\_lmt}, then the \texttt{risky\_count} is reset to zero. A sequence of identified faulty behaviour in consecutive cycles, leading to the \texttt{risky\_count} becoming equal to \texttt{persistence\_lmt} is regarded as a \textit{permanent fault}.

For each input unit, the voter unit also maintains an \texttt{isolation status}. The following are the requirements related to the isolation of a unit.
\begin{requirement} \label{isolatedifpersistence} 
A unit is isolated if and only if the \texttt{risky\_count} of the unit reaches the predefined threshold indicated by \texttt{persistence\_lmt}.
\end{requirement}
 \begin{requirement}\label{isolatedisisolated} If an input unit is isolated in a cycle of operation, it continues to remain isolated for all future cycles. 
 \end{requirement}
 \subsection{Output Generation Related Requirements}\label{subsec:output_gen}
Using the current cycle input data from each input unit and based on the updated status information of each input unit, the voter generates an output, which is used as feed to the controller. As mentioned in the output format, the feed to the controller consists of \texttt{voter\_output}, \texttt{output\_age} and \texttt{validity\_status}.

\subsubsection{Voter output and prime unit switching}
   A \texttt{unit\_output} consists of the \texttt{uid} of a unit and a \texttt{reading} received from it.
   To explain the requirements regarding the output, the following definition will be useful. 

\begin{definition}[healthy data]
    A $\mathtt{healthy\_data}$ is a $\mathtt{unit\_output}$ with $\mathtt{good}$ health such that, in the cycle of its measurement the input unit with the corresponding $\mathtt{uid}$ is not isolated and its $\mathtt{miscomparison\_status}$ is $\mathtt{not\_miscomparing}$.
\end{definition}
 The \texttt{voter\_output} in a cycle is the selected \texttt{unit\_output} and the \texttt{uid} present in it is the prime unit of the cycle. In case of a transient fault in the prime unit,  the \texttt{voter\_output} in the previous cycle will remain the \texttt{voter\_output} for current cycle, until the fault disappears or becomes permanent.  One basic requirement is that if a fault is identified, it should be contained within the faulty unit and should not be allowed to corrupt the input to the controller. This requirement is stated below.  
\begin{requirement}\label{healthy} The voter always outputs a \texttt{healthy\_data}. 
\end{requirement}
If the prime unit of the previous cycle gets isolated in the current cycle and another healthy unit is available, then the voter should select another input unit to provide the feed for the controller. However, as mentioned in the introduction, a requirement in this context is that frequent change of the prime unit used for output generation must be avoided. This requirement is formulated below. 
 \begin{requirement}\label{switched_only_if_isolated} 
 Switching of the prime unit occurs in a cycle only if the prime unit of the previous cycle got isolated. 
 \end{requirement}
\noindent The following proposition shows that if the requirements defined so far are satisfied, switching of prime unit will not happen frequently.
 \begin{proposition}\label{prop:switch_persistence}
     By satisfying the requirements R\ref{risky_count_correctness}, R\ref{isolatedifpersistence}, R\ref{healthy} and R\ref{switched_only_if_isolated}, it is ensured that the input unit selected as the prime unit shall not change for at least $\mathtt{persistence\_lmt}$ number of cycles.
 \end{proposition}    
\begin{proof}
Consider the time when an input unit $u$ is just selected as the prime unit. By R\ref{healthy}, $u$ has health status \texttt{good} and miscomparison status \texttt{not\_miscomparing}. From R\ref{switched_only_if_isolated}, a new prime unit is selected only when $u$ gets isolated. However, by requirements 
 R\ref{risky_count_correctness} and R\ref{isolatedifpersistence}, 
it must take at least \texttt{persistence\_lmt} cycles
for this to happen. 
\end{proof}

\subsubsection{Validity status related requirements}
The correctness of the voter behaviour (fault identification, isolation and output generation) is guaranteed only if the \textit{maximum permanent fault assumption} is satisfied. The voter maintains a \texttt{validity\_status} to indicate the reliability of the output generated.  The validity status 
can be \texttt{valid}, \texttt{un\_id}(unidentifiable) or \texttt{not\_valid}. 

The validity status \texttt{not\_valid} indicates that the voter output is not reliable because the \textit{maximum permanent fault assumption} does not hold and hence, no meaningful fault identification is possible. As mentioned in the introduction and in section~\ref{sec:prolongisolation}, our voter unit design does a finer fault identification, even if it is incomplete, to delay the situation of the system failing to satisfy the  
\textit{maximum permanent fault assumption}. As a consequence, there could be situations wherein the prime unit of the previous cycle is isolated in the current cycle
and there is no input unit with \texttt{healthy\_data} to be used as the new prime unit, satisfying R\ref{healthy}.  This condition occurs when the fault identification
logic is failing to identify even one input unit with \texttt{healthy\_data} because the condition mentioned in R\ref{not_miscomparing} is not satisfied.
In such scenarios, if \textit{maximum permanent fault assumption} still holds, there is a possibility that the faulty condition is transient and the voter may be able to recover from this condition in a few cycles. Hence, the voter unit can choose to retain the \texttt{voter\_output} of the previous cycle. The \texttt{un\_id} validity status is used to capture this situation. The validity status \texttt{valid} represents a higher level of reliability than \texttt{un\_id}, as we will explain shortly. The requirements regarding validity status are captured below. 
\begin{requirement} \label{not_valid} 
The validity status of the voter is \texttt{not\_valid} if and only if the system does not satisfy the \textit{maximum permanent fault assumption}.
\end{requirement}

\begin{requirement} \label{min_survive}
The validity status of the voter is  \texttt{un\_id} if and only if the \textit{maximum permanent fault assumption} is satisfied and the prime unit of the previous cycle is isolated and no input unit is identified to be providing a \texttt{healthy\_data}. 
\end{requirement} 

\begin{requirement} \label{min_for_dev_for_valid} 
If the system satisfies the \textit{maximum permanent fault assumption} and the \textit{minimum surviving units assumption} (see Definition~\ref{def:min_surviving_units}), then the validity status of the voter must be \texttt{valid}. 
\end{requirement}

\begin{requirement}\label{valid} 
If the \textit{maximum permanent fault assumption} is satisfied, and if at least one unit is identified to be providing a \texttt{healthy\_data}, then the validity status of the voter must be \texttt{valid}. 
\end{requirement}   

\subsubsection{Output age related requirements}
The \texttt{output\_age} is used by the voter to indicate the number of cycles since when the voter output measurement has not been updated. A basic correctness requirement related to \texttt{output\_age} updation is the following:
\begin{requirement}\label{req:age_correctness}
Unless the \texttt{validity\_status} is \texttt{not\_valid}, \texttt{output\_age} is either zero or one more than the previous cycle \texttt{output\_age}. Further, if validity status is \texttt{valid}, then the 
\texttt{output\_age} is equal to the  \texttt{risky\_count} of the prime unit.
\end{requirement}
An upper bound on \texttt{output\_age} is required to ensure that the feed to the controller is a measurement with a bounded lag.  When the validity status is \texttt{valid}, our requirements for the \texttt{output\_age} are as follows. 

\begin{requirement} \label{age_zero} The \texttt{output\_age} is zero if and only if the validity status is \texttt{valid} and voter output is a current cycle measurement.
\end{requirement} 

\begin{requirement} \label{age_valid_atmost}
If the validity status is \texttt{valid}, then $\mathtt{output\_age} \le \mathtt{persistence\_lmt}-1$. 
 \end{requirement} 

It is interesting to note that from the requirements stated so far, an upper bound of $2*(\mathtt{persistence\_lmt} -1)$ is guaranteed for \texttt{output\_age} when the validity status is \texttt{un\_id}. To prove this, we will use a relation that exists between the \texttt{risky\_count} of non isolated input units and the \texttt{output\_age} when the voter validity is other than \texttt{not\_valid}. The relation is stated below.

\begin{claim}\label{clm:age_risky_count_relation}
Suppose all the requirements are satisfied. Then, unless the validity status is \texttt{not\_valid}, for each input unit $u$ which is not isolated, 
$\mathtt{output\_age} - \mathtt{risky\_count} \text{ of } u   < \mathtt{persistence\_lmt}$.
\end{claim}
\begin{proof}
   We will show that this invariant is satisfied in each cycle if all the requirements (specifically, R\ref{risky_count_correctness}, R\ref{isolatedifpersistence}, R\ref{not_valid}, R\ref{min_survive}, R\ref{req:age_correctness}, and R\ref{age_valid_atmost}) are satisfied. By requirements R\ref{isolatedifpersistence} and R\ref{not_valid}, if the validity status becomes \texttt{not\_valid}, it must remain so in the future. Initially, all units have \texttt{risky\_count} zero and the \texttt{output\_age} is also zero. 
   Since $\mathtt{persistence\_lmt} > 0$, the relation holds initially for all input units.  
   Now, consider a cycle in which the validity status is different from \texttt{not\_valid} after the cycle's fault identification. If the validity status is \texttt{valid}, then by requirement R\ref{age_valid_atmost}, $\mathtt{output\_age} < \mathtt{persistence\_lmt}$ and hence the claim holds. Otherwise, the validity status must be \texttt{un\_id}. In this case, consider an arbitrary input unit $u$ which remains non-isolated after the cycle. It suffices to show that if the relation holds for $u$ before the cycle, it holds after the cycle as well. 
   
   By R\ref{min_survive}, $u$ cannot be providing a \texttt{healthy\_data} in the current cycle. Therefore, by requirement R\ref{risky_count_correctness}, the \texttt{risky\_count} of $u$ must be exactly one more than its value in the previous cycle. Moreover, by R\ref{req:age_correctness}, $\mathtt{output\_age}$ can increase by at most one. Hence, as required, if the relation holds for $u$ before the cycle, it holds after the cycle as well.
\end{proof}
 The following proposition is now easy to prove. 
\begin{proposition}\label{prop:age_not_valid_atmost} 
    If the system satisfies requirements R\ref{risky_count_correctness} to R\ref{age_valid_atmost}, then the $\mathtt{output\_age}$ is at most $2*(\mathtt{persistence\_lmt} -1)$ as long as the \textit{maximum permanent fault assumption} holds.
\end{proposition}
\begin{proof}
By requirement R\ref{not_valid}, unless the validity status is \texttt{not\_valid}, the maximum permanent fault hypothesis holds and so, there are at least $\mathtt{min\_required} > 0$ non-isolated units. Consider one such non-isolated unit $u$. By Claim~\ref{clm:age_risky_count_relation}, $\mathtt{output\_age} - \mathtt{risky\_count} \text{ of } u   \le \mathtt{persistence\_lmt} - 1$. Further, by R\ref{isolatedifpersistence}, the \texttt{risky\_count} of $u$ is at most $\mathtt{persistence\_lmt}-1$.  Hence, $\mathtt{output\_age}  < 2*(\mathtt{persistence\_lmt} -1)$.     
\end{proof}
 \begin{note}
 It is not obvious that the requirements R\ref{risky_count_correctness} to R\ref{age_valid_atmost} are consistent. The consistency is established by 
 machine assisted proof created using Rocq. En route, we also obtained machine generated proofs of the first part of Claim~\ref{clm:partial_identification_miscomparing}, the first part of Claim~\ref{clm:clm:partial_identification_not_miscomparing}, Proposition~\ref{prop:completeness_achieved},  Proposition~\ref{prop:switch_persistence} and  Claim~\ref{clm:age_risky_count_relation}. 
\end{note}
\section{Design of the Voter Unit}\label{sec:design} 
\begin{figure}
 \centering
 \resizebox{\textwidth}{!}{  \begin{tikzpicture}[shorten >=1pt, node distance=13cm, on grid, auto]
    \node[state, initial, minimum size=4cm] (init)   {\parbox{4.2cm}{ \centering {\LARGE $\mathbf{s_0}$}\\ {\Large $\mathtt{isolated\_units =0}$ \\ \texttt{age = 0} \\ \texttt{valid}}}}; 
     
    \node[state, minimum size=3.2cm, below=of init] (validagenonzero)   {\parbox{3.25cm}{ \centering {\LARGE $\mathbf{s_1}$} {\Large\\\texttt{0 < age < persistence\_lmt} \\  \texttt{ valid}}}}; 
    
     \node[state,  minimum size=1cm, right=of validagenonzero] (validagezero)   {\parbox{3.65cm}{  \centering{\LARGE $\mathbf{s_2}$}\\ {\Large$\mathtt{isolated\_units > 0}$ \\ \texttt{age = 0} \\ \texttt{valid}}}};

    \node[state, minimum size=4cm, below=of validagenonzero] (un_id) {\parbox{4.3cm}{  \centering {\LARGE $\mathbf{s_3}$}\\ {\Large\texttt{0 < age < 2*persistence\_lmt-1} \\
     \texttt{un\_id }}}}; 
     
    \node[state, minimum size=4cm, right=of un_id] (not_valid) {\parbox{5cm}{ \centering {\LARGE $\mathbf{s_4}$}\\ {\Large\texttt{age =2*persistence\_lmt} \\
     \texttt{not\_valid}}}}; 
    
    \path[->] 
   (init) edge [loop above] node {\parbox{5.5cm}{ \Large \texttt{prime\_unit is healthy and isolated\_units = 0?\\: update output} }} ()

   (init) edge [bend left=20] node[pos=0.3] {\parbox{6cm}{ \Large{\texttt{prime\_unit is 
    unhealthy?\\: incr age}} }}(validagenonzero)
 
    (validagezero) edge [bend left=10] node[pos=0.45] {\parbox{6.2cm}{\Large{\texttt{prime\_unit is 
    unhealthy\\and not\_isolated?: incr age}}}}(validagenonzero)
    
    (validagenonzero) edge [loop left] node [sloped, pos=0.86] {\parbox{4.1cm}{{\Large $\mathtt{non\_isolated\_units}$ $\mathtt{\ge  min\_required}$ \\ \texttt{prime\_unit is unhealthy and not\_isolated?\\: incr age} }}} ()

     (validagenonzero) edge [bend left=20] node {\parbox{5.2cm}{\Large\texttt{prime\_unit is 
    healthy and }$\mathtt{isolated\_units =0}$\texttt{?\\: reset \texttt{age},\\update output}} }(init)

    (validagenonzero) edge [->, sloped, rotate around={0.1:(validagezero)}] node {\parbox{6cm}{\Large \texttt{prime\_unit is 
    healthy and} $\mathtt{isolated\_units > 0}$ \texttt{?\\ : reset \texttt{age}, update output}}}(validagezero)
     
    (validagenonzero) edge [bend left =60 ] node[pos=0.52]{\parbox{8cm}{\Large \texttt{previous prime\_unit is isolated} and $\mathtt{non\_isolated\_units \ge min\_required}$ \texttt{and} $\mathtt{healthy\_list \ne nil}$ \texttt{?: switch prime\_unit, reset age, update output}}}(validagezero)

 (validagezero) edge [loop right, pos=0.45] node {\parbox{2.2cm}{\Large \texttt{prime\_unit is healthy ?: update output}} } () 

 (validagenonzero) edge [bend left =15] node [pos=0.7] {\parbox{4.1cm}{ \Large $\mathtt{non\_isolated\_units}$ $\mathtt{< min\_required?}$\\\texttt{: set age to 2*persistence\_lmt} }} (not_valid)
 
    (un_id) edge [->, sloped, rotate around={0.1:(validagezero)}] node {\parbox{4.6cm} {\Large {$\mathtt{non\_isolated\_units <}$ $\mathtt{min\_required}$\texttt{?\\: set age to 2*persistence\_lmt}}}} (not_valid)

     (validagenonzero) edge [->, left] node  {\parbox{5cm}{\Large $\mathtt{non\_isolated\_units \ge}$ \texttt{min\_required and} $\mathtt{healthy\_list = nil}$ \texttt{and
  \\  previous prime\_unit is isolated?\\: incr age }}}(un_id)
      
 (un_id) edge [bend left =3, sloped] node [pos=0.3] {\parbox{4.6cm}{\raggedright \Large
    $\mathtt{non\_isolated\_units\ge}$ $\mathtt{min\_required }$ \texttt{and}
     $\mathtt{healthy\_list \ne nil}$\texttt{?\\
    : switch \textit{prime\_unit}, update output}}} (validagezero)

 (un_id) edge [loop below] node {\parbox{7.4cm}{\Large$\mathtt{non\_isolated\_units \ge min\_required}$ 
 \texttt{ and } $\mathtt{healthy\_list = nil}$
 \texttt{ and \\ previous prime\_unit is isolated?\\: incr age}}} ()

 (not_valid) edge [loop below] node {} ()

  (init) edge [bend left=90] node {\parbox{6cm}{\Large {\texttt{prime\_unit is healhy but another unit isolated?\\ : update output}}} }(validagezero)

  (validagezero) edge [bend left] node{\parbox{5cm} {\raggedright \Large {\texttt{prime\_unit is healthy but} $\mathtt{non\_isolated\_units}$ $ < \mathtt{min\_required}$ \texttt{?\\: set age to} $\mathtt{2*persistence\_lmt}$}}}  (not_valid)
        
    ;

\end{tikzpicture}}
\caption{State transition diagram of the voter unit. To avoid cluttering, the update in the isolation status of units and the number of isolated units during the transitions are not shown. The update in validity status is not explicitly shown during the transitions, but can be understood from the states.}
 \label{fig:voter_validity}
\end{figure}
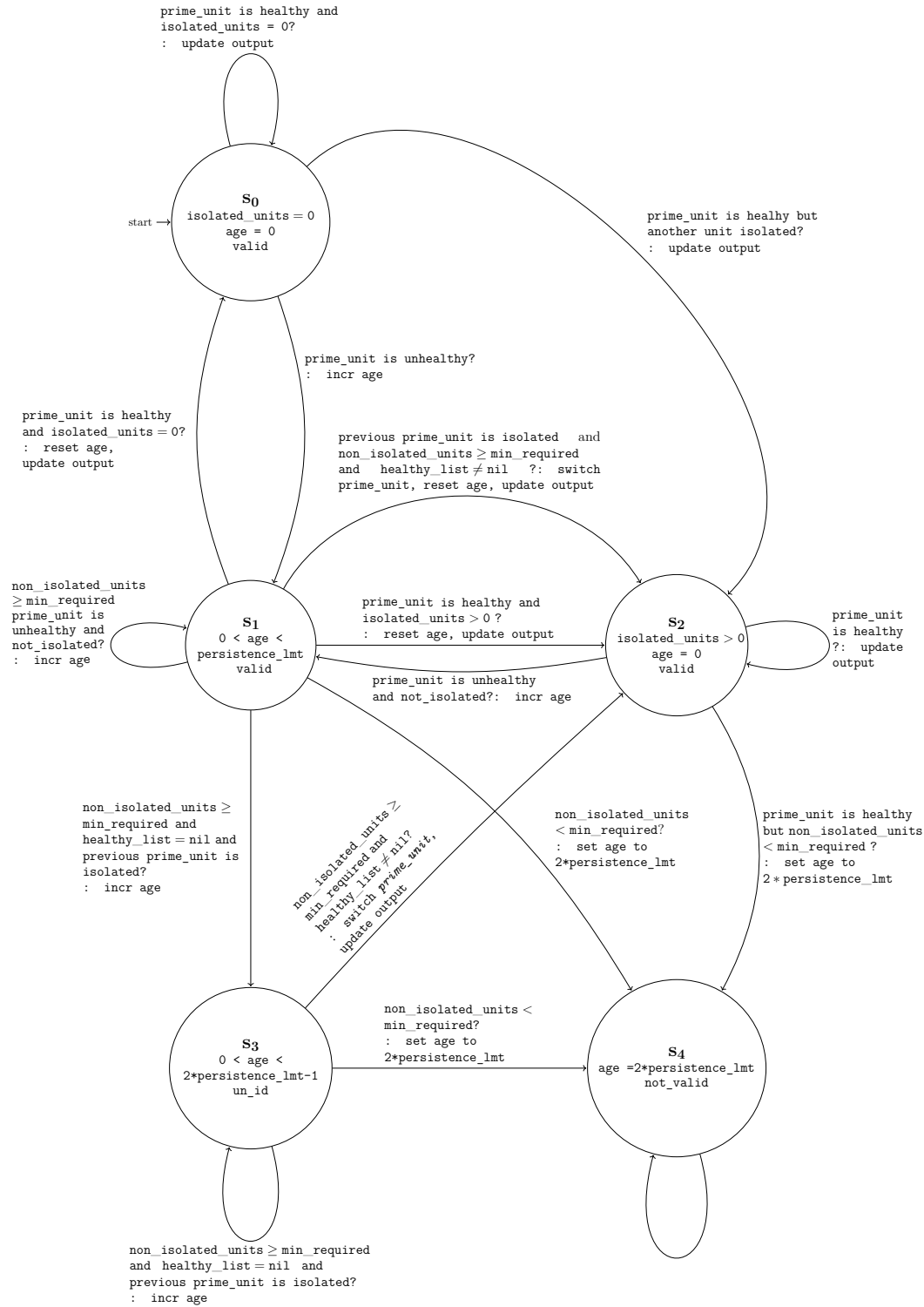

As described in the voter output format, the voter unit provides a \texttt{voter\_output} for the controller along with an indication of the \texttt{validity\_status} and \texttt{output\_age}. The design of the voter unit has three key aspects: fault identification, handling of transient and permanent faults and computation of the voter output. The voter unit is designed as a state transition system. The configurable parameters of the voter unit are:  \texttt{num\_units} to indicate the count of input units, $\delta$ to indicate the noise threshold, \texttt{persistence\_lmt} used for isolation of permanent faulty units and \texttt{max\_simul\_fault} to indicate the maximum number of simultaneous faults that are anticipated in a cycle of operation. 

 In every transition, the voter unit will update the fault status (\texttt{unit\_status}) of each input unit by performing fault identification on the current cycle reading. The \texttt{unit\_status} includes the \texttt{isolation\_status}, \texttt{miscomparison\_status} and the \texttt{risky\_count} of the corresponding input unit.
 The logic used in the fault identification algorithm is an exact translation of the conditions specified in requirements R\ref{miscomparing} and R\ref{not_miscomparing}. If both conditions do not apply, then the unit will be classified as \texttt{maybe\_miscomparing}. This method will satisfy requirements R\ref{maybe_nil}~-~R\ref{not_miscomparing}, as explained in Section~\ref{sec:soundandcomplete} and \ref{sec:prolongisolation}.
After the fault identification, the update in \texttt{risky\_count} is done as described in R\ref{risky_count_correctness}. Using this, the isolation of permanently faulty units is done as mentioned in requirements R\ref{isolatedifpersistence} and R\ref{isolatedisisolated}.

We explain the voter unit behaviour using five abstract states, as shown in the state transition diagram in Figure~\ref{fig:voter_validity}. The description of the states are as follows.

\begin{itemize}
    \item $S_0$: This is the initial state where there are no isolated inputs, the \texttt{output\_age} is zero and the \texttt{validity\_status} is \texttt{valid}. The output value will be the current cycle measurement of the prime unit.
    \item  $S_1$: This state occurs when the prime unit used for \texttt{voter\_output} generation is having a transient fault. The previous cycle \texttt{voter\_output} is retained to satisfy R\ref{healthy}, and \texttt{output\_age} is incremented on transition to this state. 
    \item $S_2$: This state occurs when the number of isolated units is non-zero, but non-isolated units are at least \texttt{min\_required}. There is at least one unit providing \texttt{healthy\_data} and one such unit is selected as the prime unit. Further, the \texttt{output\_age} is zero, indicating that the output value is a current cycle measurement from prime unit, and \texttt{validity\_status} is \texttt{valid}.
    \item $S_3$: This state occurs when the input unit used for \texttt{voter\_output} generation in the previous cycle got isolated and there are no input units providing \texttt{healthy\_data}, even though the number of non-isolated units is at least \texttt{min\_required}. This state satisfies the premise of the requirement R\ref{min_survive} and hence \texttt{validity\_status} is \texttt{un\_id}. The previous cycle \texttt{voter\_output} is retained to satisfy R\ref{healthy}, and \texttt{output\_age} is incremented on transition to this state.
    \item $S_4$: This state occurs only when the number of non-isolated units is below \texttt{min\_required}. In this state, the validity status is $\texttt{not\_valid}$ as per requirement R\ref{not_valid}. 
\end{itemize}
\begin{note}
     It can be verified from the description of states that the requirements R\ref{isolatedifpersistence}, R\ref{healthy}, R\ref{not_valid}, R\ref{min_survive}, R\ref{valid}, R\ref{age_zero}, R\ref{age_valid_atmost} and the second part of requirement R\ref{req:age_correctness} are invariants for all states. 
\end{note}

The state transitions are summarized below.
\begin{itemize}
    \item From the initial state $S_0$, there are three transitions. These are a) to state $S_1$, when the prime unit used for \texttt{voter\_output} generation is having a transient fault and b) to state $S_2$, when the prime unit is healthy and one of the other units is isolated. The system remains in state $S_0$ as long as the initially selected prime unit is healthy and none of the other units are isolated.
    \item From states $S_1, S_2$ and $S_3$, a transition to $S_4$ happens if and only if the \textit{maximum permanent fault assumption} is not satisfied due to isolation of input units. 
    \item From state $S_1$, there are five more transitions. These transitions are a) to state $S_1$ when the prime unit has a transient fault b) to state $S_2$ when the prime unit of the previous cycle is isolated and there is at least one healthy unit to be identified as the prime unit c) to state $S_2$ when the prime unit of the previous cycle recovers from a transient fault and at least one input unit is isolated d) to state $S_0$ when prime unit of previous cycle recovers from a transient fault and none of the input units are isolated e) to state $S_3$ when prime unit of previous cycle is isolated and there are no healthy units to be identified as a prime unit.
    \item From state $S_2$, there is a transition to the same state when the prime unit is healthy. There is also a transition to state $S_1$  when the selected prime unit has a transient fault.
    \item From state $S_3$, there are two  more transitions. These are a) to state $S_3$ when there are no healthy units and b) to state $S_2$ when one of the input units recovers from a transient fault and becomes the prime unit.
    \item Once the system reaches the state $S_4$, it should remain there, in order to satisfy the requirements R\ref{isolatedisisolated} and R\ref{not_valid}.
\end{itemize}
\begin{note}
 It can be seen that the requirements R\ref{risky_count_correctness}~- R\ref{not_miscomparing}, R\ref{isolatedisisolated},  R\ref{switched_only_if_isolated}, R\ref{min_for_dev_for_valid}, and the first part of requirement R\ref{req:age_correctness} are invariants of the state transition rules. 
 \end{note}

\section{Implementation in Rocq}\label{sec:implementationinrocq}
We have used the design given in Section~\ref{sec:design} for our 
Rocq implementation. We have used Rocq's support for higher-order logics and dependent types extensively (see Appendix~\ref{sec:implementation_details}) while formally expressing the requirements mentioned in Section~\ref{sec:functionalrequirements}. 
Our implementation follows a correct-by-construction approach by defining the voter state and the state transition rules in such a way that they carry with them the proof terms of the invariants they satisfy. This makes sure that when these definitions are completed and are verified by Rocq, the invariants will certainly hold by the implementation. 
Since Rocq is based on the Calculus of Construction, all proofs are fully constructable, rather than being existential. The total time taken for verification in Rocq for properties R1 to R16 is approximately 23 seconds on a laptop with Debian 12 OS, 32GB RAM, 64-Bit Processor (Intel® Core™ Ultra 5 125U × 14).

The verified code obtained in Rocq is extractable to OCaml/Haskell/Scheme using the provisions in Rocq~\cite{RocqWebcite2}.  The number of lines of the extracted OCaml code is 362.  The execution time of the extracted OCaml code for a single cycle state transition is around 3 milliseconds.

To get an overview of our implementation using Rocq, the reader may refer to Appendix~\ref{sec:implementation_details}. The definition of the voter unit is given as the record type \texttt{voter\_state}. It may be noticed that the definition of \texttt{voter\_state} carries embedded proof terms of properties corresponding to requirements 
R\ref{healthy}, R\ref{not_valid}, R\ref{min_survive}, R\ref{valid}, R\ref{age_zero}, R\ref{age_valid_atmost} and the second part of requirement R\ref{req:age_correctness}. The requirement R\ref{isolatedifpersistence} is an invariant of \texttt{unit\_status}, which is a field inside the \texttt{voter\_state}. The the state transition function is named as \texttt{voter\_state\_transition}. It carries proof terms of properties corresponding to requirements  R\ref{risky_count_correctness}~- R\ref{not_miscomparing}, R\ref{isolatedisisolated}, R\ref{switched_only_if_isolated}, 
R\ref{min_for_dev_for_valid}, and the first part of requirement R\ref{req:age_correctness}. 
The entire source code of our implementation is available in GitHub\footnote{\url{https://github.com/arifali-ap/Fault_tolerant_voter.git}}.

\section{Discussion and Future Directions}\label{sec:futuredirections}
Our work gives the design and synthesized code of a voter unit that feeds a control system. The synthesized code is correct by construction. The design is parameterized by the number of units and some reliability parameters.  The verified code for the parameterized system is produced in Rocq. The values of the parameters can be appropriately fixed by the user during extraction of the code from Rocq to OCaml/Haskell/Scheme. The support available in Rocq for  higher order logics and dependent types helped in presenting the formal specification of our requirements in a clear and succinct way.  

In fault tolerant systems with a multi-level hierarchical design, redundancy management is required in each level of hierarchy. An extension in this direction is interesting. Further, one may have to deal with fault handling of asynchronous redundant units in such systems.  
For example, in a distributed fault-tolerant system, the control unit may receive feed from multiple asynchronous voting units, each working with independent logical clocks with small deviations.  Each of these voting units may be receiving data from redundant input units. Synthesis of verified code for identifying and handling faults in such asynchronous systems is another direction for future work.

%


\begin{thebibliography}{10}

\bibitem{Aviziens1976Fault}
A.~Aviziens.
\newblock Fault-tolerant systems.
\newblock {\em IEEE Transactions on Computers}, C-25(12):1304--1312, 1976.
\newblock \href {https://doi.org/10.1109/TC.1976.1674598}
  {\path{doi:10.1109/TC.1976.1674598}}.

\bibitem{Butler1993Formal}
Ricky Butler and Sally Johnson.
\newblock {\em Formal Methods For Life-Critical Software}, pages 319--329.
\newblock Aerospace Research Central, October 1993.
\newblock \href {https://doi.org/10.2514/6.1993-4516}
  {\path{doi:10.2514/6.1993-4516}}.

\bibitem{Chlipala2013certified}
Adam Chlipala.
\newblock {\em Certified Programming with Dependent Types: A Pragmatic
  Introduction to the {C}oq Proof Assistant}.
\newblock The MIT Press, 2013.

\bibitem{HandbookofModelchecking}
Edmund~M. Clarke, Thomas~A. Henzinger, Helmut Veith, and Roderick Bloem,
  editors.
\newblock {\em Handbook of Model Checking}.
\newblock Springer, 2018.
\newblock \href {https://doi.org/10.1007/978-3-319-10575-8}
  {\path{doi:10.1007/978-3-319-10575-8}}.

\bibitem{dajani2003formal}
Samar Dajani-Brown, Darren Cofer, Gary Hartmann, and Steve Pratt.
\newblock Formal modeling and analysis of an avionics triplex sensor voter.
\newblock In {\em Model Checking Software}, pages 34--48, Berlin, Heidelberg,
  2003. Springer Berlin Heidelberg.

\bibitem{wakerly1978synchronization}
Daniel Davies and John~F. Wakerly.
\newblock Synchronization and matching in redundant systems.
\newblock {\em IEEE Transactions on Computers}, C-27:531--539, 1978.
\newblock URL: \url{https://api.semanticscholar.org/CorpusID:15222328}.

\bibitem{Davis1987Davis}
J~Gloria Davis.
\newblock An analysis of redundancy management algorithms for asynchronous
  fault tolerant control systems.
\newblock Technical report, Ames Research Center, California, NASA, 1987.

\bibitem{ADIS1675datasheet}
Analog Devices.
\newblock {ADIS}1675 {D}atasheet.
\newblock Accessed: 2025-07-10.
\newblock URL: \url{https://www.analog.com/en/products/adis16575.html}.

\bibitem{divito1990formal}
Ben~L DiVito, Ricky~W Butler, and James~L Caldwell.
\newblock Formal design and verification of a reliable computing platform for
  real-time control. {P}hase 1: {R}esults.
\newblock Technical report, NASA Langley Research Center, 1990.

\bibitem{goldberg1984}
Jack Goldberg, Milton~W. Green, William H. Kautz Karl~N. Levitt, P.~Michael
  Melliar-Smith, Richard~L. Schwartz, and Charles~B. Weinstoc.
\newblock Development and analysis of the software implemented fault-tolerance
  (sift) computer.
\newblock Technical report, SRI International, Menlo Park, California, 1984.

\bibitem{Hitt2006Fault}
E.F. Hitt and D.~Mulcare.
\newblock Fault-tolerant avionics.
\newblock In C.R. Spitzer and Spitzer, editors, {\em Digital Avionics Handbook
  (1st ed.)}, chapter~28, pages 1--28. CRC Press, 2006.
\newblock \href {https://doi.org/https://doi.org/10.1201/9781420036879}
  {\path{doi:https://doi.org/10.1201/9781420036879}}.

\bibitem{GG1320usermanual}
Honeywell.
\newblock {GG1320 U}ser manual.
\newblock Accessed: 2025-07-10.
\newblock URL:
  \url{https://aerospace.honeywell.com/content/dam/aerobt/en/documents/learn/products/sensors/user-manuals/GG1320-usermanual.pdf}.

\bibitem{jones2017modular}
Benjamin~F. Jones and Lee Pike.
\newblock Modular model-checking of a byzantine fault-tolerant protocol.
\newblock In {\em NASA Formal Methods}, pages 163--177, Cham, 2017. Springer
  International Publishing.

\bibitem{Kassab2014Anovel}
M.~A. Kassab, H.~S. Taha, S.~A. Shedied, and A.~Maher.
\newblock A novel voting algorithm for redundant aircraft sensors.
\newblock In {\em Proceeding of the 11th World Congress on Intelligent Control
  and Automation}, pages 3741--3746, 2014.
\newblock \href {https://doi.org/10.1109/WCICA.2014.7053339}
  {\path{doi:10.1109/WCICA.2014.7053339}}.

\bibitem{Knight2002Safety}
J.C. Knight.
\newblock Safety critical systems: challenges and directions.
\newblock In {\em Proceedings of the 24th International Conference on Software
  Engineering. ICSE 2002}, pages 547--550, 2002.

\bibitem{krishnan2024formalclk}
Ranjani Krishnan, Ashutosh Gupta, Nitin Chandrachoodan, and V~R Lalithambika.
\newblock Formal verification of clock synchronization algorithms for aerospace
  systems.
\newblock In {\em 2024 8th International Conference on Electronics,
  Communication and Aerospace Technology (ICECA)}, pages 1603--1608, 2024.
\newblock \href {https://doi.org/10.1109/ICECA63461.2024.10800832}
  {\path{doi:10.1109/ICECA63461.2024.10800832}}.

\bibitem{krishnan2024formal}
Ranjani Krishnan, Ashutosh Gupta, Nitin Chandrachoodan, and VR~Lalithambika.
\newblock Formal verification of voting algorithms for safety critical systems
  using two approaches.
\newblock In {\em 2024 IEEE Space, Aerospace and Defence Conference (SPACE)},
  pages 211--215. IEEE, 2024.

\bibitem{lamport1982byzantine}
Leslie Lamport, Robert Shostak, and Marshall Pease.
\newblock The {B}yzantine generals problem.
\newblock {\em ACM Trans. Program. Lang. Syst.}, 4(3):382–401, July 1982.
\newblock \href {https://doi.org/10.1145/357172.357176}
  {\path{doi:10.1145/357172.357176}}.

\bibitem{lincoln1993formal}
Patrick Lincoln and John Rushby.
\newblock The formal verification of an algorithm for interactive consistency
  under a hybrid fault model.
\newblock In {\em Computer Aided Verification}, pages 292--304, Berlin,
  Heidelberg, 1993. Springer Berlin Heidelberg.

\bibitem{Miller2003Proving}
Steven~P. Miller, Alan~C. Tribble, and Mats P.~E. Heimdahl.
\newblock Proving the shalls.
\newblock In Keijiro Araki, Stefania Gnesi, and Dino Mandrioli, editors, {\em
  FME 2003: Formal Methods}, pages 75--93, Berlin, Heidelberg, 2003. Springer
  Berlin Heidelberg.

\bibitem{miner1993verification}
Paul~S Miner.
\newblock Verification of fault-tolerant clock synchronization systems.
\newblock Technical report, NASA Langley Research Center, 1993.

\bibitem{Softerrorbook2011}
M.~Nicolaidis.
\newblock {\em Soft Errors In Modern Electronic Systems}.
\newblock Springer New York, NY, 01 2011.
\newblock \href {https://doi.org/10.1007/978-1-4419-6993-4}
  {\path{doi:10.1007/978-1-4419-6993-4}}.

\bibitem{Owre1995Formal}
S.~Owre, J.~Rushby, N.~Shankar, and F.~von Henke.
\newblock Formal verification for fault-tolerant architectures: prolegomena to
  the design of {PVS}.
\newblock {\em IEEE Transactions on Software Engineering}, 21(2):107--125,
  1995.
\newblock \href {https://doi.org/10.1109/32.345827}
  {\path{doi:10.1109/32.345827}}.

\bibitem{Pike2006Anote}
L.~Pike.
\newblock A note on inconsistent axioms in rushby's "systematic formal
  verification for fault-tolerant time-triggered algorithms".
\newblock {\em IEEE Transactions on Software Engineering}, 32(5):347--348,
  2006.
\newblock \href {https://doi.org/10.1109/TSE.2006.41}
  {\path{doi:10.1109/TSE.2006.41}}.

\bibitem{Powell1992Failuremode}
D.~Powell.
\newblock Failure mode assumptions and assumption coverage.
\newblock In {\em [1992] Digest of Papers. FTCS-22: The Twenty-Second
  International Symposium on Fault-Tolerant Computing}, pages 386--395, 1992.
\newblock \href {https://doi.org/10.1109/FTCS.1992.243562}
  {\path{doi:10.1109/FTCS.1992.243562}}.

\bibitem{rahli2018velisarios}
Vincent Rahli, Ivana Vukotic, Marcus V{\"o}lp, and Paulo Esteves-Verissimo.
\newblock Velisarios: Byzantine fault-tolerant protocols powered by {C}oq.
\newblock In {\em Programming Languages and Systems}, Cham, 2018. Springer
  International Publishing.

\bibitem{rushby1999systematic}
J.~Rushby.
\newblock Systematic formal verification for fault-tolerant time-triggered
  algorithms.
\newblock {\em IEEE Transactions on Software Engineering}, 25(5):651--660,
  1999.
\newblock \href {https://doi.org/10.1109/32.815324}
  {\path{doi:10.1109/32.815324}}.

\bibitem{Rushby1993Formal}
J.M. Rushby and F.~von Henke.
\newblock Formal verification of algorithms for critical systems.
\newblock {\em IEEE Transactions on Software Engineering}, 19(1):13--23, 1993.
\newblock \href {https://doi.org/10.1109/32.210304}
  {\path{doi:10.1109/32.210304}}.

\bibitem{rushby1994formally}
John Rushby.
\newblock A formally verified algorithm for clock synchronization under a
  hybrid fault model.
\newblock In {\em Proceedings of the Thirteenth Annual ACM Symposium on
  Principles of Distributed Computing}, PODC '94, page 304–313, New York, NY,
  USA, 1994. Association for Computing Machinery.
\newblock \href {https://doi.org/10.1145/197917.198115}
  {\path{doi:10.1145/197917.198115}}.

\bibitem{rushby2002overview}
John Rushby.
\newblock An overview of formal verification for the time-triggered
  architecture.
\newblock In {\em Formal Techniques in Real-Time and Fault-Tolerant Systems},
  pages 83--105, Berlin, Heidelberg, 2002. Springer Berlin Heidelberg.

\bibitem{rushby1989formal}
John Rushby and Frieder VonHenke.
\newblock Formal verification of a fault tolerant clock synchronization
  algorithm.
\newblock Technical report, NASA Langley Research Center, 1989.

\bibitem{saha2016formal}
Indranil Saha, Suman Roy, and S.~Ramesh.
\newblock Formal verification of fault-tolerant startup algorithms for
  time-triggered architectures: A survey.
\newblock {\em Proceedings of the IEEE}, 104:1--19, 03 2016.
\newblock \href {https://doi.org/10.1109/JPROC.2016.2519247}
  {\path{doi:10.1109/JPROC.2016.2519247}}.

\bibitem{RocqWebcite2}
Rocq~Core Team.
\newblock Program extraction.
\newblock Accessed: 2025-07-08.
\newblock URL:
  \url{https://rocq-prover.org/doc/master/refman/addendum/extraction.html}.

\bibitem{RocqWebcite}
Rocq~Core Team.
\newblock Rocq prover.
\newblock Accessed: 2025-07-08.
\newblock URL: \url{https://rocq-prover.org/}.

\end{thebibliography}
\newpage
\appendix
\section{Implementation Details}\label{sec:implementation_details}
In this section we describe the details of our implementation of the voter's operational logic as per the design in the Section~\ref{sec:design}. 

\subsection{Basic Data Types and their Invariants}\label{subsec:basic_data_types}
As mentioned in Section~\ref{sec:systemdescription}, each input unit is assumed to have a \texttt{unit\_id}. This is defined as below. 
\begin{verbatim}
	Inductive unit_id := 
	uid_con : nat -> unit_id.
\end{verbatim}

The  \texttt{signal\_health} type given below is used to represent the self-declared health status of input unit (\texttt{good} or \texttt{bad}).

\begin{verbatim}
Inductive signal_health :=
    | bad  : signal_health
    | good : signal_health.
\end{verbatim}

The \texttt{signal} type given below is used to represent measurement from an input unit (\texttt{val} $\in \mathbb{N)}$ along with a health status (\texttt{hw\_hlth}).
\begin{verbatim}
Record signal := signal_build {
    val     : nat;
    hw_hlth : signal_health;
}.
\end{verbatim}

A \texttt{reading} of type \texttt{signal} and \texttt{uid} of type \texttt{unit\_id} together constitute a \texttt{unit\_output} data type.  As mentioned in Section~\ref{sec:systemdescription}, this represents the information received by the voter from each input unit.
\begin{verbatim}
Record unit_output := unit_output_build {
    reading : signal;    
    uid     : unit_id; 
}.
\end{verbatim}

As explained in the functional requirement, an accumulated status of the fault identification and isolation of each input unit is maintained by the voter. This is represented  using the data type \texttt{unit\_status} defined below. 

\begin{verbatim}
Record unit_status := unit_status_build {
    iso_status     : isolation;
    miscomp_status : miscomparison_status;
    risky_count    : nat;
    pf_risky_count : risky_count <= persistence_lmt /\ 
                    (risky_count = persistence_lmt <-> 
                                            iso_status = isolated);
	}.   
\end{verbatim}

The proof term \texttt{pf\_risky\_count} is included as a field in \texttt{unit\_status} so as to satisfy the requirement R\ref{isolatedifpersistence} as an invariant.

The voter sees the  data type \texttt{unit\_data} from each input unit. This is a record of unit\_output and unit\_status as defined below. 

\begin{verbatim}
Record unit_data := unit_data_build { 
    u_output   : unit_output; 
    u_status   : unit_status;
    pf_healthy : (risky_count u_status) = 0 
                 <-> ( hw_hlth (reading u_output) = good
                    /\(iso_status u_status)     = not_isolated
                    /\(miscomp_status u_status) = not_miscomparing 
                    );
}.  
\end{verbatim}

The proof term \texttt{pf\_healthy} carries the invariant that the \texttt{risky\_count} is zero if and only if the data of the input unit is a \texttt{healthy\_data}.

\subsection{Voter State and its Invariants}\label{subsec:state_invariants}
As mentioned in Section~\ref{sec:systemdescription}, the voter receives input from a predefined number of input units. In our implementation, this is represented using a parameter \texttt{num\_units} $\in \mathbb{N}$.
To indicate that the voter sees inputs from units with unique ids, we define a data type \texttt{uid\_list} which is a list of \texttt{unit\_id},  with an embedded proof that there is no duplication in the list.
\begin{verbatim}
	Record uid_list := uid_list_build {
		     l    : list unit_id;
		     pf_l : NoDup l;
	}.
\end{verbatim}
A list named \texttt{u\_ids} of type \texttt{uid\_list} is the concrete list of  \texttt{unit\_id} defined for the system in which the unit ids correspond to  the sequence of numbers from $1$ to \texttt{num\_units}.

The voter is implemented as a state transition system. States are represented using the \texttt{voter\_state} data type which is defined below.
\begin{verbatim}
Record voter_state := voter_state_build {
    u_data_lst    : list unit_data;
    (* On unit switching, the next unit is selected as per the 
       order in u_data_lst *)
    pf_ud_lst     : get_u_ids_of_snsr_data u_data_lst = u_ids.(l);
    voter_output  : unit_output;
    voter_validity: validity_status;
    output_age    : nat;                          
    presrvd_data  : unit_data;
    pf_v_output   : In (uid voter_output)
                            (get_u_ids_of_snsr_data u_data_lst);
    pf_presrvd_data : presrvd_data.(u_output) = voter_output
                            /\ healthy_data presrvd_data;
    pf_age        : output_age = 0 <->
                                (voter_validity = valid
                                /\ In presrvd_data u_data_lst);
    pf_validity : pf_validity_prop 
                          voter_output voter_validity u_data_lst;
    pf_out_not_isolated : voter_validity = valid ->
                            In (uid voter_output) 
                                (get_u_ids_of_snsr_data
                                non_isolated_list u_data_lst));
    pf_risky_cnt    : risky_cnt_age_prop pf_ud_lst
                                         pf_v_output
                                         output_age
                                         voter_validity;
    pf_age_bound   : voter_validity <> not_valid -> 
                    forall x, 
                     In x u_data_lst
                     -> x.(u_status).(iso_status) = not_isolated 
                     -> output_age - x.(u_status).(risky_count)
                             < persistence_lmt;
    pf_age_validity : (output_age < persistence_lmt 
                              <-> voter_validity = valid)
                      /\ (output_age >= 2*persistence_lmt 
                             <-> voter_validity = not_valid)
}.                
\end{verbatim}

The important fields of \texttt{voter\_state} are as follows. 
\begin{itemize}
    \item \texttt{u\_data\_lst}, which is a list of \texttt{unit\_data}, with one element corresponding to each input unit, to keep the readings and status of the corresponding unit. The proof term \texttt{pf\_ud\_lst} indicates that the list of unit ids of the input units in the \texttt{u\_data\_lst} is the list \texttt{u\_ids}. 
    \item \texttt{voter\_output} of type \texttt{unit\_output} that is assumed to be fed to the controller. The proof term \texttt{pf\_v\_output} states that the \texttt{voter\_output} of the voter is a measurement from one of the input units whose unit id is in the \texttt{u\_data\_lst}.
    \item \texttt{voter\_validity} of type validity\_status that is assumed to be fed to the controller
    \item \texttt{output\_age} $\in \mathbb{N}$ that is assumed to be fed to the controller.
    \item \texttt{presrvd\_data},  which is the \texttt{unit\_data} of the prime unit selected for voter output generation. This information is required by the voter to decide the state transition.  
\end{itemize}

The \texttt{voter\_state} definition also has the following proof terms specifying invariants to satisfy some of the functional requirements described in Section~\ref{sec:functionalrequirements}. The requirement R\ref{isolatedifpersistence} is an invariant of \texttt{unit\_status}, which is defined in Section~\ref{subsec:basic_data_types}. The requirements R\ref{healthy}, R\ref{not_valid}, R\ref{min_survive}, R\ref{valid},R\ref{age_zero}, R\ref{age_valid_atmost} and the second part of requirement R\ref{req:age_correctness} are embedded as invariants using the proof terms in the \texttt{voter\_state} definition.

\begin{itemize}
	 \item  The requirement R\ref{healthy} is captured by the proof term \texttt{pf\_presrvd}.
    
    \item \texttt{pf\_validity} is a direct translation of the requirements R\ref{not_valid}, R\ref{min_survive} and R\ref{valid} relating the fields \texttt{voter\_output}, \texttt{voter\_validity} and the \texttt{u\_data\_lst}. \\   
The translation of  R\ref{not_valid} is done as follows. 
\begin{verbatim}
(count_of_non_isolated_units u_data_lst  < min_required  
      <->  voter_validity = not_valid )
\end{verbatim}

\noindent The translation of  R\ref{min_survive} is as follows.
\begin{verbatim}
/\(( min_required <= count_of_non_isolated_units u_data_lst
    /\ In (uid voter_output) 
     (get_u_ids_of_unit_data (isolated_list u_data_lst))) 
    <-> voter_validity = un_id)
/\ ( voter_validity = un_id -> 
            (healthy_unit_list u_data_lst) = nil)
\end{verbatim}    

\noindent The translation of  R\ref{valid} is as follows.
\begin{verbatim}  
/\ (( min_required <= count_of_non_isolated_units u_data_lst
/\ ((healthy_unit_list u_data_lst) <> nil) )
                    -> voter_validity = valid ).
\end{verbatim}
\item The second part of requirement R\ref{req:age_correctness} is captured by the proof term \texttt{pf\_risky\_cnt}. 
   
The translation of the second part of requirement R\ref{req:age_correctness} is done as follows. 
\begin{verbatim}
voter_validity = valid ->
  let voter_out_data := find_data_of_a_given_unit pf_ud_lst pf_in in
  output_age =  ( risky_count (u_status (proj1_sig voter_out_data)))
\end{verbatim} 

\item   The requirement R\ref{age_zero} is captured by the proof term \texttt{pf\_age} in combination with the proof term \texttt{pf\_presrvd}.

\item The requirement R\ref{age_valid_atmost} and Proposition~\ref{prop:age_not_valid_atmost} are captured by the proof term \texttt{pf\_age\_validity}. 
      
\item \texttt{pf\_age\_bound} is a direct translation of Claim~\ref{clm:age_risky_count_relation} and it is used for proving Proposition~\ref{prop:age_not_valid_atmost}.      
    \item \texttt{pf\_out\_not\_isolated} satisfies the invariant that if the  \texttt{voter\_validity} is \texttt{valid}, then the prime unit is not isolated. This proof term helps in proving that  R\ref{min_survive} is maintained during a state transition. 
	
\end{itemize}

\subsection{Voter State Update Rules and its Properties}\label{subsec:transition_invariants}
Among the requirements, those related to a single \texttt{voter\_state} are embedded in the voter state definition as explained in the previous section. The remaining requirements are about the invariants to be satisfied during each state transition. 

The state transition in each cycle is done by using the \texttt{voter\_state\_transition} function, which takes the previous cycle voter state \texttt{vs} and the current cycle \texttt{unit\_output} of each input unit and generates an updated voter state \texttt{new\_vs}.  The function is designed in such a way that its output is the new voter state along with a proof term \texttt{voter\_state\_transition\_prop} which encodes the remaining requirements. The \texttt{voter\_state\_transition} function has two main steps.
\begin{enumerate}
    \item The \texttt{build\_updated\_u\_data\_lst} function takes as input the \texttt{u\_data\_lst} (which is a list of \texttt{unit\_data}) of the previous cycle voter state \texttt{vs} and the current cycle \texttt{unit\_output} of all input units. It creates a new list of \texttt{unit\_data} denoted by \texttt{new\_p\_ud\_lst} in which the \texttt{unit\_status} of each input unit is updated as per the fault identification and isolation algorithm.
    \item The \texttt{unit\_status} of input units in \texttt{new\_p\_ud\_lst} is used as the \texttt{u\_data\_lst} of the new voter state \texttt{new\_vs}. Based on the updated isolation status of input units available in \texttt{new\_p\_ud\_lst}, the assessment of whether the system satisfies \textit{maximum permanent fault assumption} is done. Based on this, the remaining fields of the new voter state \texttt{new\_vs} are generated.
\end{enumerate}
\noindent We will now describe the details of the steps given above. 
\paragraph{Creating the updated \texttt{unit\_data} list \texttt{new\_p\_ud\_lst}. }

To begin with, fault identification has to be done. For this, a filtered list
     of \texttt{unit\_output}, namely\\
     \texttt{all\_good\_non\_iso\_lst}, is created based on the previous cycle voter state \texttt{vs} and the list of current cycle \texttt{unit\_output}. 

This is done so as to filter out units previously isolated as per \texttt{vs} and units with self identifying health status \texttt{bad} in the current cycle, so that they are not used for identifying deviation faults. From \texttt{all\_good\_non\_iso\_lst}, two sub-lists of unit ids are computed. The list of unit ids which satisfy the premise of requirement R\ref{miscomparing} are identified as \texttt{dev\_uid\_lst}. The list of unit ids which are not included in \texttt{dev\_uid\_lst} and which do not satisfy the premise of requirement R\ref{not_miscomparing} are identified as \texttt{maybe\_uid\_lst}.

The next step is to use the lists computed above to build the updated list of \texttt{unit\_data}, called \texttt{new\_p\_ud\_lst}. The current cycle \texttt{unit\_output} available as input to\\ \texttt{build\_updated\_u\_data\_lst} is copied to the \texttt{unit\_output} of the corresponding element of \texttt{new\_p\_ud\_lst}. Now, the \texttt{unit\_status} of each input unit of \texttt{new\_p\_ud\_lst} needs to be computed from the \texttt{unit\_status} of the corresponding input unit of the \texttt{u\_data\_lst} of the previous cycle voter state \texttt{vs}. This is done as follows. 

If the isolation status in the \texttt{u\_data\_lst} of \texttt{vs} is \texttt{isolated}, the isolation status in \texttt{new\_p\_ud\_lst} is \texttt{isolated}. The other fields in \texttt{unit\_status} of \texttt{new\_p\_ud\_lst} remain as in \texttt{u\_data\_lst} of \texttt{vs}. If the isolation status in \texttt{u\_data\_lst} of \texttt{vs} is \texttt{not\_isolated}, then the following rules are applied.
         \begin{itemize}
         \item If the self identified health status in current cycle \texttt{unit\_output} is  \texttt{bad}, then \texttt{risky\_count} in \texttt{new\_p\_ud\_lst} is one more than that in \texttt{u\_data\_lst} of the voter state \texttt{vs}. 
        
         \item If the self identified health status in current cycle \texttt{unit\_output} is \texttt{good} and if the \texttt{uid} is included in \texttt{dev\_uid\_lst} (resp.~in \texttt{maybe\_uid\_lst}), the miscomparison status in \texttt{new\_p\_ud\_lst} is set to \texttt{miscomparing} (resp. is set to ~\texttt{maybe\_miscomparing}) and the \texttt{risky\_count} in \texttt{new\_p\_ud\_lst} is one more than that in \texttt{u\_data\_lst} of \texttt{vs}. This is required to meet the second part of the requirement R\ref{risky_count_correctness}.
         
        \item  If the self identified health status in current cycle \texttt{unit\_output} is \texttt{good} and if the \texttt{uid}  is not included in \texttt{dev\_uid\_lst} and in \texttt{maybe\_uid\_lst},  the miscomparison status in \texttt{new\_p\_ud\_lst} is \texttt{not\_miscomparing} and the \texttt{risky\_count} in \texttt{new\_p\_ud\_lst} is set to zero.
        \item  When the \texttt{risky\_count} after updation reaches \texttt{persistence\_lmt}, the isolation status of that unit in \texttt{new\_p\_ud\_lst} is  set to \texttt{isolated}. 
         \end{itemize} 
         
It is proved that the above rules of creation of \texttt{new\_p\_ud\_lst} guarantee that requirements R\ref{risky_count_correctness} to R\ref{not_miscomparing} and R\ref{isolatedisisolated} are satisfied by the fault identification and isolation algorithm. Since \texttt{new\_p\_ud\_lst} is used as the \texttt{u\_data\_lst} of \texttt{new\_vs}, these requirements will also be satisfied by the voter state transition.

\paragraph{Voter Output Generation using \texttt{new\_p\_ud\_lst}. }
The new voter state \texttt{new\_vs} uses  \texttt{new\_p\_ud\_lst} created above as its \texttt{u\_data\_lst}. The other fields of \texttt{new\_vs}, such as \texttt{validity\_status}, \texttt{voter\_output} and \texttt{output\_age} are defined based on the updated information available in \texttt{new\_p\_ud\_lst}.

Using the count of units with isolation status as \texttt{not\_isolated} in \texttt{new\_p\_ud\_lst}, the decision whether the system satisfies the \textit{maximum permanent fault assumption} is taken. If this count
is less than \texttt{min\_required}, then as per requirement R\ref{not_valid}, \texttt{new\_vs} has  \texttt{validity\_status} as \texttt{not\_valid}. In this case, the \texttt{output\_age} is set as \texttt{2*persistence\_lmt} and the remaining fields are unchanged from those of previous \texttt{voter\_state} \texttt{vs}.  If the count is greater than or equal to \texttt{min\_required}, the rules used are the following. 

As per \texttt{new\_p\_ud\_lst}, if the new isolation status 
    of the unit used in the \texttt{voter\_output} of the previous state  \texttt{vs} is \texttt{isolated}, then a healthy data list is computed from \texttt{new\_p\_ud\_lst}, which is the list of its elements with isolation status as \texttt{not\_isolated}, self identified health status \texttt{good} and miscomparison status as \texttt{not\_miscomparing}. 
     \begin{itemize}
       \item If the healthy data list is non-empty, then to be consistent with the requirement R\ref{valid},  the \texttt{validity\_status} of \texttt{new\_vs} is set to \texttt{valid}. The \texttt{unit\_output} and \texttt{unit\_data} from the first element in the list are respectively used as   
       \texttt{voter\_output} and \texttt{presrvd\_data} of \texttt{new\_vs}. The \texttt{output\_age} is set as zero. 
        \item If the healthy data list is empty, then to be consistent with the requirement R\ref{min_survive}, the \texttt{validity\_status} of \texttt{new\_vs} is set to \texttt{un\_id}. The \texttt{output\_age} is one more than that of \texttt{vs}. The remaining fields are unchanged from those of previous \texttt{voter\_state} \texttt{vs}.
\end{itemize}

     As per \texttt{new\_p\_ud\_lst}, if the new isolation status 
    of the unit used in the \texttt{voter\_output} of the previous state  voter state \texttt{vs} is \texttt{not\_isolated}, then the \texttt{validity\_status} of \texttt{new\_vs} is set as \texttt{valid}, to be consistent with the requirement R\ref{valid}. 

    \begin{itemize}
    \item If as per \texttt{new\_p\_ud\_lst}, that unit is providing a healthy data, then the \texttt{voter\_output} and \texttt{presrvd\_data} of \texttt{new\_vs} are build using the \texttt{unit\_data} of that unit, available in \texttt{new\_p\_ud\_lst} and the \texttt{output\_age} is set to zero. 
    \item Otherwise, the \texttt{voter\_output} and \texttt{presrvd\_data} of \texttt{new\_vs} are same as they were in the previous voter state \texttt{vs}. The \texttt{output\_age} is one more than that of the previous voter state \texttt{vs}.
    \end{itemize}

\paragraph{\textbf{Invariants Maintained during State Transitions.}}
As mentioned earlier, the state transition function is \texttt{voter\_state\_transition}. It takes the previous cycle voter state \texttt{vs} and the current cycle \texttt{unit\_output} of each input unit and generates a new voter state \texttt{new\_vs}, which satisfies a proof term \texttt{voter\_state\_transition\_prop}. The property \texttt{voter\_state\_transition\_prop} is an encoding of requirements R\ref{risky_count_correctness} -  R\ref{not_miscomparing},
 R\ref{isolatedisisolated}, R\ref{switched_only_if_isolated}  and R\ref{min_for_dev_for_valid} and first part of requirement R\ref{req:age_correctness}.

 We define a property \texttt{simul\_fault\_prop}, which is the translation of the \textit{Simultaneous Fault Hypothesis} of the system. The property means  the following:\\
 The number of non-isolated units with \texttt{bad} health $+$ the number of non-isolated units with \texttt{good} health and having at least $\delta$ deviation from ground truth is  $\le \mathtt{simul\_max\_fault}$.

\noindent The requirement R\ref{risky_count_correctness} is  translated as below.
\begin{verbatim}
 (  forall x, In x (u_data_lst vs)
    -> forall y, In y (u_data_lst new_vs)
    -> (uid (u_output x) = uid (u_output y))
    -> risky_cnt_prop x y )
\end{verbatim}

where, the property \texttt{risky\_cnt\_prop} is defined as the following.
\begin{verbatim}
Definition risky_cnt_prop   (old : unit_data) (new :  unit_data) :=
  iso_status old.(u_status) = not_isolated ->
  (risky_count new.(u_status) = 0  
  \/ risky_count new.(u_status) = S (risky_count ( old).(u_status)))
        /\ ( (risky_count new.(u_status) = S(risky_count(old).(u_status))) 
             <-> (miscomp_status (u_status new) <> not_miscomparing 
                 \/ hw_hlth (reading (u_output new)) = bad)    )
     
\end{verbatim}
\noindent The first part of requirement R\ref{maybe_nil} is translated as given below.
\begin{verbatim}
(simul_fault_prop (pf_l u_ids)  pf_all_unit_outputs (pf_ud_lst vs)
    -> let mis_flt_lmt := flt_lmt_among_good
             u_ids.(pf_l) pf_all_unit_outputs (pf_ud_lst vs) in
    let l := proj1_sig (all_good_non_iso_lst
             u_ids.(pf_l) pf_all_unit_outputs (pf_ud_lst vs) ) in
        length l > 2*mis_flt_lmt
            -> forall x, In x (u_data_lst vs)
            -> forall y, In y (u_data_lst new_vs)
            -> (uid (u_output x) = uid (u_output y))                   
            -> iso_status (u_status x) = not_isolated
            -> y.(u_status).(miscomp_status) <> maybe_miscomparing )
\end{verbatim}

\noindent The completeness part of requirement R\ref{maybe_nil} is translated as below.
\begin{verbatim}
( simul_fault_prop (pf_l u_ids)  pf_all_unit_outputs (pf_ud_lst vs)
    -> min_required <= count_of_non_isolated_units vs.(u_data_lst)
    -> forall x, In x (u_data_lst vs)
    -> forall y, In y (u_data_lst new_vs)
    ->  forall z, In z all_unit_outputs
    -> (uid (u_output x) = uid (u_output y))
    -> (uid (u_output x) = uid z)  
    -> iso_status (u_status x) = not_isolated
    
    ->( let mis_flt_lmt := flt_lmt_among_good
    u_ids.(pf_l) pf_all_unit_outputs (pf_ud_lst vs) in
    let l := proj1_sig(all_good_non_iso_lst
    u_ids.(pf_l) pf_all_unit_outputs (pf_ud_lst vs) ) in
    length l > 2*mis_flt_lmt
    
     -> z.(reading).(hw_hlth) = good 
     -> adiff ground_truth z.(reading).(val) > 3*delta
     -> y.(u_status).(miscomp_status) = miscomparing )
   
    /\( let mis_flt_lmt := flt_lmt_among_good
    u_ids.(pf_l) pf_all_unit_outputs (pf_ud_lst vs) in
    let l := proj1_sig (all_good_non_iso_lst
     u_ids.(pf_l) pf_all_unit_outputs (pf_ud_lst vs) ) in
    length l > 2*mis_flt_lmt
    
    -> In z l
    -> adiff ground_truth z.(reading).(val) <= delta
    -> y.(u_status).(miscomp_status) = not_miscomparing)

\end{verbatim}

\noindent The requirement R\ref{always_sound} is translated as below.
\begin{verbatim}
( simul_fault_prop (pf_l u_ids)  pf_all_unit_outputs (pf_ud_lst vs)
    -> min_required <= count_of_non_isolated_units vs.(u_data_lst))
    -> forall x,In x (u_data_lst vs)
    -> forall y,In y (u_data_lst new_vs)
    ->  forall z, In z all_unit_outputs
    -> (uid (u_output x) = uid (u_output y))
    -> (uid (u_output x) = uid z)  
    -> iso_status (u_status x) = not_isolated
    -> ( (* soundness a prop *)
     y.(u_status).(miscomp_status) = miscomparing
    -> adiff ground_truth z.(reading).(val)    > delta )
    
       (* soundness B prop *)
    /\  ( let  l := proj1_sig (all_good_non_iso_lst (pf_l u_ids)
        pf_all_unit_outputs (pf_ud_lst vs)) in
    In z l -> y.(u_status).(miscomp_status) = not_miscomparing
    -> adiff ground_truth z.(reading).(val) <= 3*delta )
\end{verbatim}

The requirement R\ref{miscomparing} is translated as below.
\begin{verbatim}
    In x (u_data_lst vs)
    -> forall y,In y (u_data_lst new_vs)
    ->  forall z,In z all_unit_outputs
    -> (uid (u_output x) = uid (u_output y))
    -> (uid (u_output x) = uid z)  
    -> iso_status (u_status x) = not_isolated
    -> 
  ( let mis_flt_lmt := flt_lmt_among_good
    u_ids.(pf_l) pf_all_unit_outputs (pf_ud_lst vs) in
    let l := proj1_sig (all_good_non_iso_lst
    u_ids.(pf_l) pf_all_unit_outputs (pf_ud_lst vs) ) in
    In z l
    -> miscomparing_many_check l mis_flt_lmt z = true
    -> y.(u_status).(miscomp_status) = miscomparing )
\end{verbatim}

\indent The requirement R\ref{not_miscomparing} is translated as below.
\begin{verbatim}
    In x (u_data_lst vs)
    -> forall y, In y (u_data_lst new_vs)
    ->  forall z, In z all_unit_outputs
    -> (uid (u_output x) = uid (u_output y))
    -> (uid (u_output x) = uid z)  
    -> iso_status (u_status x) = not_isolated
    -> (let mis_lst  := 
                proj1_sig(miscomparing_lst
                (pf_l u_ids) pf_all_unit_outputs(pf_ud_lst vs) ) in
        let mis_flt_lmt := flt_lmt_among_good  (pf_l u_ids)
                            pf_all_unit_outputs(pf_ud_lst vs) in
        let rem_mis_flt_lmt  := mis_flt_lmt - length ( mis_lst ) in
        let gd_non_iso_lst   := proj1_sig (all_good_non_iso_lst
                    (pf_l u_ids)pf_all_unit_outputs(pf_ud_lst vs) )in
        let negb_mis_lst     := filter(fun y =>
        negb (miscomparing_many_check 
                gd_non_iso_lst mis_flt_lmt y ) )gd_non_iso_lst in 
        In z negb_mis_lst
         -> agreeing_many_check negb_mis_lst rem_mis_flt_lmt z = true
         -> y.(u_status).(miscomp_status) = not_miscomparing )
   ).
\end{verbatim}
\noindent The requirement R\ref{isolatedisisolated} is translated as below.
\begin{verbatim}
  ( forall x, In x (get_u_ids_of_unit_data   
              (isolated_list (u_data_lst vs))) -> 
              In x (get_u_ids_of_unit_data   
                        (isolated_list (u_data_lst new_vs))))
\end{verbatim}
\noindent The requirement R\ref{switched_only_if_isolated}  is described here. 
\begin{verbatim}
 ( (uid (voter_output vs)) <> (uid (voter_output new_vs))
    -> In (uid(voter_output vs)) 
    (get_u_ids_of_unit_data (isolated_list (u_data_lst new_vs))))
\end{verbatim}
\noindent The requirement R\ref{min_for_dev_for_valid} is translated as below.
\begin{verbatim}
  (simul_fault_prop (pf_l u_ids)  pf_all_unit_outputs (pf_ud_lst vs)
     -> count_of_non_isolated_units (u_data_lst new_vs) >= min_required
     -> let mis_flt_lmt := flt_lmt_among_good
     u_ids.(pf_l) pf_all_unit_outputs (pf_ud_lst vs) in
     let l := proj1_sig (all_good_non_iso_lst 
            u_ids.(pf_l) pf_all_unit_outputs (pf_ud_lst vs) ) in
     length l > 2*mis_flt_lmt
     -> voter_validity new_vs = valid )
\end{verbatim}

\noindent The first part of requirement R\ref{req:age_correctness} is proved using the proof term described. 
\begin{verbatim}
  ( voter_validity new_vs <> not_valid
   -> (  output_age new_vs = 
      S (output_age vs) \/ output_age new_vs = 0 ))
\end{verbatim}
\noindent It may be recalled that the other requirements are maintained as invariants of \texttt{voter\_state}, as mentioned in Section~\ref{subsec:state_invariants}.  
\end{document}